\documentclass[acmsmall]{acmart}

\setcopyright{none}
\renewcommand\footnotetextcopyrightpermission[1]{} 
\usepackage{booktabs} 

\usepackage[ruled]{algorithm2e} 

\usepackage{stmaryrd}
\usepackage[normalem]{ulem}
\usepackage{enumitem}
\usepackage{cancel}

\SetAlFnt{\small}
\SetAlCapFnt{\small}
\SetAlCapNameFnt{\small}
\SetAlCapHSkip{0pt}
\IncMargin{-\parindent}
\newcommand{\E}{\mathbb{E}}

\newtheorem*{remark}{Remark}

\usepackage{arydshln}

\usepackage{subcaption}
\usepackage[normalem]{ulem}
\acmJournal{TWEB}
\acmVolume{?}
\acmNumber{?}
\acmArticle{?}
\acmYear{2018}
\acmMonth{5}
\copyrightyear{?}

\setcopyright{acmlicensed}

\begin{document}

\author{Tim Hellemans, Grzegorz Kielanski and Benny Van Houdt}
\affiliation{%
  \institution{University Of Antwerp}
  \streetaddress{Middelheimlaan 1}
  \city{Antwerp}
  \postcode{2000}
  \country{Belgium}}
  
  \title[]{Performance of Load Balancers with Bounded Maximum Queue Length in case of Non-Exponential Job Sizes}

\begin{abstract}
In large-scale distributed systems, balancing the load in an efficient way is crucial in order to achieve low latency. Recently, some load balancing policies have been suggested which are able to achieve a bounded maximum queue length in the large-scale limit. However, these policies have thus far only been studied in case of exponential job sizes. As job sizes are more variable in real systems, we investigate how the performance of these policies (and in particular the value of these bounds) is impacted by the job size distribution.

We present a unified analysis which can be used to compute the bound on the queue length 
in case of phase-type distributed job sizes for four 
load balancing policies. We find that in most cases, the bound on the maximum queue length can be expressed in closed form. In addition, we obtain job size (in)dependent bounds on the expected response time. 

Our methodology relies on the use of the cavity process. That is, we conjecture that the cavity process captures the behaviour of the real system as the system size grows large. For each policy, we illustrate the accuracy of the cavity process by means of simulation.
\end{abstract}

\fancyfoot{}
\maketitle
\thispagestyle{empty}

\section{Introduction}
Load balancing plays a crucial role in achieving low latency in
large-scale clusters. A well studied family of load balancing policies is referred to as \textit{power-of-$d$-choices} load balancing policies. Two of the most prominent examples of this family are the SQ($d$) policy (see e.g.~\cite{mitzenmacher2, vvedenskaya3}) and the LL($d$) policy (see e.g.~\cite{hellemansSIG18}). One of the main theoretical insights from the study of these policies is the sharp decay of the tail of the queue length distribution (see e.g.~\cite{bramsonAAP}). 
More recently, several load balancing policies have been introduced that achieve a
bounded maximum queue length in the large-scale limit (i.e., as the number of servers
tends to infinity). We can distinguish $3$ seminal papers 
in this area that are considered in this paper.

Hyperscalable load balancing was introduced in \cite{van2019hyper} (and further studied in \cite{van2021optimal, zhou2021asymptotically}). For this policy the dispatcher maintains
an upper bound on the actual queue lengths, updates these bounds at random times and assigns jobs greedily.
In this paper we additionally consider an analogous setting where servers initiate the updates
instead of the dispatcher.

In \cite{ying2017power} the authors consider a policy that reduces the overhead of SQ$(d)$
by gathering incoming jobs in large batches and assigning these batches in a water-filling
manner to a large set of randomly selected servers. We show that the analysis for this model is closely related to the analysis of the hyperscalable policy mentioned above.

In \cite{tsitsiklis2013power} the authors study the power of (even a little) resource pooling. This means that a fraction of the total processing power is centralized in a single (fast) server. This fast server is then configured to always steal work from the server which currently has the longest queue. While the methodology used to compute the bound on the queue length is similar for this model, its analysis in general is significantly different from the other models. 

In \cite{van2019hyper,ying2017power,tsitsiklis2013power}  simple expressions for the bounded maximum queue length in the large-scale limit were presented
for exponential job sizes. In contrast, it is well known that in real systems the job size distribution is much more variable. Often, a significant part of the total workload is offered by a small fraction of long jobs, while the remaining workload consists mostly of short jobs (e.g.~\cite{delgado2016job, delgado2015hawk, Sparrow}). Therefore, we focus on job sizes which have a phase-type distribution (further denoted by PH distribution). PH distributions are distributions with a modulating finite state Markov chain (see also \cite{latouche1}). While many of our results also apply for general job sizes (as indicated in the text), we keep our focus on PH distributions for ease of presentation. Moreover, any general positive-valued distribution can be approximated arbitrarily close with a PH distribution and there are various fitting tools available for PH distributions (see e.g.~\cite{Kriege2014,panchenko1}).

Our analysis relies on the methodology of the queue at the cavity \cite{bramsonLB}. The queue at the cavity is used to approximate the large system behavior and is equivalent to determining the unique fixed point of a fluid approximation. The idea is that as the number of servers grows large, the state (which includes the queue length) of the servers become independent and identically distributed. Therefore, when the number of servers is sufficiently large, the performance of the whole system can be well approximated by studying a single queue, which is called the queue at the cavity. While for exponential job sizes the convergence towards a fluid approximation was proven for the water filling and resource pooling policies in \cite{ying2017power, tsitsiklis2012power}, this was already highly challenging due to the discontinuities in
the drifts. Moreover, proving that the cavity method yields exact asymptotic results
for more general job size distributions is hard (see e.g.~\cite{bramsonLB_QUESTA}), often
due to the lack of monotonicity. Therefore, we focus on the analysis of the cavity queue and assume that it yields exact results as the number of servers tends
to infinity. Simulation experiments are presented in Section 
\ref{app:simulation} which support this assumption.

The rest of this work is structured as follows. In Section \ref{sec:model_description} we give a formal definition of all models we consider throughout this work. In Section \ref{sec:main_results} we briefly discuss the main results we obtained. In Sections \ref{sec:push}-\ref{sec:pooling} we give a detailed study of each policy, here we also provide further insights through analytical and numerical experimentation. We conclude in Section \ref{sec:future_work} and indicate future work directions.
	
\section{Model Description} \label{sec:model_description}
While we are considering $4$ separate models, it is worthwhile to introduce them all at once as their model descriptions have many commonalities. We consider a system with $N$ homogeneous servers which all process jobs at a constant rate equal to one (resp.~$1-p$ for resource pooling). Jobs arrive to a central dispatcher according to a Poisson process with arrival rate $\lambda \cdot N$. We assume that the size of a job has a PH distribution with parameters $(\alpha, S)$. Furthermore, we use the notation $n_s = |\alpha|$ 
to denote the number of phases and let $s^* = - S \textbf{1}$ with $\textbf{1}$ an $n_s \times 1$ column consisting of ones. Without loss of generality, we assume the mean job size is equal to one. Furthermore, for all policies, whenever a tie occurs of any sort, these are broken uniformly at random. Under this setting, we now distinguish $4$ distinct policies/models:
\begin{itemize}
	\item For the \textit{push policy} \cite{van2019hyper}, we assume there is some $\delta > 0$ such that the dispatcher probes a random server at a rate equal to $\delta N$. Whenever a server is probed, its queue length is saved at the dispatcher, this estimated queue length is then incremented by one whenever the dispatcher assigns a job to this queue. For each incoming job, the dispatcher assigns the job to a server which has the lowest estimated queue length.
	\item The \textit{pull policy} is similar to the push policy in the sense that the dispatcher keeps track of estimated queue lengths and assigns incoming jobs to the server with the smallest estimated queue length. However, queue length updates are now sent by the servers. Whenever a server finishes a job it sends its queue length to the dispatcher with probability $\delta_1$. Furthermore, when a server is idle it sends an update to the dispatcher at rate $\delta_0$,
	where $\delta_1$ and $\delta_0$ are such that the overall probe rate equals $\delta$.
	\item For the \textit{water filling policy} \cite{ying2017power}, jobs arrive at the dispatcher in batches (or are aggregated) which consist of $M$ tasks. The arrival process is a Poisson process with rate $\frac{N}{M} \lambda$, the batch size $M$ is assumed to be of order $\Theta(\log(N))$ and increasing as a function of $N$.
	
	Given probe rate $\delta > 0$, each batch of jobs selects $\frac{\delta}{\lambda} M$ queues and the $M$ jobs are assigned using \textit{water filling}. That is, the $M$ tasks are added one by one to the $\frac{\delta}{\lambda} M$ servers by assigning each job in the batch to the server with the shortest queue amongst the $\frac{\delta}{\lambda} M$ selected servers. E.g.~when $N= M = 3$ and the queue lengths are given by $(0,1,4)$, the queue lengths are increased to $(2,2,4)$ by one batch arrival.
	\item For the \textit{resource pooling policy} \cite{tsitsiklis2013power}, incoming jobs join the queue of a random server. There is also an additional parameter $p$ which signifies the fraction of centralized service. Each individual server works at a rate equal to $1-p$ while a central server steals a job of the server with the most jobs in its queue. More specifically, the centralized server generates tokens 
	at rate $pN$ and when a token is generated, it 
	instantaneously serves a job from one of the individual servers
	with the most number of jobs in its queue.
	The job selected by the centralized server is a pending job, unless
	there are no pending jobs.
\end{itemize}

Whenever we refer to a quantity related to the push policy we add a superscript $^\shortrightarrow$, for pull a superscript $^\shortleftarrow$, for water-filling 
a superscript $^w$ and finally for resource pooling superscript $^r$.

\begin{remark}
	As the total processing rate is equal to $N$ for all considered policies, each policy remains stable for all $\lambda < 1$, while being unstable for $\lambda \geq 1$. Therefore, we let $\lambda  \in [0,1)$ throughout the text.
\end{remark}

\section{Main Results} \label{sec:main_results}
For all considered policies we develop an analytical or numerical method which can be used to efficiently compute the stationary queue length (and response time) distribution
of the queue at the cavity. The accuracy of these policies is verified in Section \ref{app:simulation}. Furthermore, we have many additional analytical results which we summarize here. To this end, let $Z$ denote the job size distribution and $X$ an exponential random variable with rate $\delta$. We find that many of our results can be stated as a function of the probability that a job finishes service before the exponential timer (with rate $\delta$) expires:
\begin{equation} \label{eq:y}
	y = P[Z < X] = \alpha (\delta I - S)^{-1} s^*.
\end{equation}
We first compute a value $\tilde m \in [0,\infty)$ such that the maximum queue length 
of the queue at the cavity is given by $\lceil \tilde m \rceil$.
For the \textit{push policy}, we show that for any job size distribution, the maximum queue length depends only on the job size distribution via $y$ and is given by:
$$
\lceil \tilde m \rceil = \left \lceil \frac{\log\left[ \frac{1}{y} + \left( \frac{\lambda}{\delta (1-\lambda)} - 1 \right) \cdot \frac{1-y}{y} \right]}{\log(1/y)} \right \rceil .
$$
From this it is easy to see that we have vanishing waiting times when $\frac{\lambda}{1-\lambda} \leq \delta$ irrespective of the job size distribution.
Moreover, we are able to derive accurate bounds on the mean queue length $E[Q^{a \shortrightarrow}]$ given by:
$$
\lfloor \tilde m \rfloor - \lambda^\shortrightarrow_{\lfloor \tilde m \rfloor} / \delta \leq E[Q^{a \shortrightarrow}] \leq \lceil \tilde m \rceil - \lambda^\shortrightarrow_{\lceil\tilde m\rceil}/\delta.
$$
Here $\lambda^\shortrightarrow_m$ denotes the arrival rate at which the maximum queue length jumps from $m$ to $m+1$, this value is given by:
$$
\lambda^{\shortrightarrow}_m = \frac{\delta y (1-y^m)}{\delta y (1-y^m)+y^m(1-y)}.
$$
Furthermore, by noting that the value of $\tilde m$ is monotone in $y$ and $y \in [e^{-\delta}, 1]$, we can let $y$ tend to $e^{-\delta}$ and $1$. This way, we establish a tight upper and lower bound on the maximal queue length:
$$
\left\lceil \frac{1}{\delta} \log \left( 
1 +  \frac{1}{\delta}   \frac{ \lambda}{1-\lambda} (e^\delta-1)\right) \right\rceil
\leq 
\lceil \tilde m \rceil \leq 
\left\lceil \frac{\lambda}{(1-\lambda)\delta} \right\rceil.
$$
While the upper bound scales as $1/(1-\lambda)$, we find that for any fixed distribution the value of $\lceil \tilde m\rceil$ scales as $\log(1/(1-\lambda))$. 

We show that the water filling model coincides with the push policy for integer values of $\tilde m$. This allows one to show that all aforementioned results for the push policy also apply for the water filling policy. In addition to these results, we also find that there is a near closed form expression of the stationary distribution for this policy.

For the pull policy, we show that the maximum queue length is \textit{insensitive} to the job size distribution, and that it is given by (with $\delta = (1-\lambda)\delta_0 +
\lambda \delta_1)$:
$$
\lceil \tilde m \rceil = \lceil \log(1-\lambda \delta_1 / \delta) / \log(1-\delta_1) \rceil.
$$
From this, we can see that we have vanishing waiting times whenever $\lambda \leq \delta$. Moreover, it is easy to see that the value of $\lceil \tilde m \rceil$ is increasing as a function of $\delta_1$, in the extreme case of $\delta_1 = 0$ we find that the maximum queue length is given by $\lceil \lambda / \delta \rceil$ which remains bounded for any value of $\lambda$. For the pull policy, we find that the maximum queue length jumps up from $m$ to $m+1$ at the arrival rate:
$$
\lambda_m^\shortleftarrow = \frac{\delta_0 - \delta_0 (1-\delta_1)^m}{\delta_0 - \delta_0 (1-\delta_1)^m + \delta_1(1-\delta_1)^m}.
$$
Furthermore, we again obtain a similar bound on the mean queue length in Theorem \ref{th:bounds_EQ_pull}. 

For \textit{resource pooling} we find that the maximum queue length does depend on the complete job size distribution in a non-trivial way and provide an efficient numerical method to compute the maximum queue length and stationary distribution. However,  closed form expressions 
appear to only be feasible for exponential job sizes. We distinguish $3$ cases for the values of $\lambda$ and $p$:
\begin{itemize}
	\item When $\lambda \leq p$, the centralized server is able to finish all incoming work. In this case, the servers are always idle and there is no queueing.
	\item When $\frac{1}{2} \left( 1 + \lambda - \sqrt{1+ 2\lambda - 3 \lambda^2} \right) \leq p < \lambda$, all servers have at most one job in their queue and we find that the queue length distribution is insensitive to the job size distribution.
	\item When $p < \frac{1}{2} \left( 1 + \lambda - \sqrt{1+ 2\lambda - 3 \lambda^2} \right)$, the queue length depends on the job size distribution, the maximum queue length can be made arbitrarily large (by increasing the variability of the job sizes) and is lower bounded when job sizes are deterministic.
\end{itemize}

We now provide the analysis and numerical insights into each of the introduced models. As the methodology is similar for all considered policies, we provide all details for the \textit{push policy} in Section \ref{sec:push}, while we might skim over some subtleties for the other policies.

\section{Hyperscalable push policy} \label{sec:push}
In this section we study the queue at the cavity for the push policy with PH distributed job sizes. The accuracy of the queue at the cavity for the push policy is demonstrated 
by simulation in Section \ref{app:simulation}. As jobs are assigned in a greedy manner based on the estimated queue lengths, we find that
in the large-scale limit, all servers have an estimated queue length equal to $m$ or $m+1$ for some integer $m \geq 0$. As the estimated queue length is an upper bound on the actual queue length, the state of the queue at the cavity can be denoted as $(q,e,j)$, where $e \in \{m,m+1\}$
is the estimated queue length, $q \in \{0,1,\ldots,e\}$ is the actual queue length and $j \in \{1,\ldots,n_s\}$ is the service phase provided that $q >0$. When $q=0$, we can simply denote
the state as $(0,e)$. In other words, the queue at the cavity has 
\[ \Omega^\shortrightarrow = \{(0,m),(0,m+1)\} \cup \{(q,e,j) | e = m,m+1 ; q = 1,\ldots,e; j = 1,\ldots,n_s \},\]
as its state space.

\subsection{State transitions}\label{sec:trans}

By definition of the matrix $S$, entry $(j,j')$ of $S$ represents the rate at which
the state changes from $(q,e,j)$ to $(q,e,j')$ due to service phase changes.
Furthermore, $s^*_j$ is the service completion
rate in phase $j$, for $j=1,\ldots,n_s$. Thus, from state $(q,e,j)$ a jump occurs to state
$(q-1,e,j')$ at rate $s^*_j \alpha_{j'}$ if $q > 1$, as a service completion occurs at rate
$s_j^*$ and a new job starts service in phase $j'$ with probability $\alpha_{j'}$. Similarly
a jump occurs from state $(1,e,j)$ to state $(0,e)$ at rate $s_j^*$. 

The state can also change due to a probe event, the queue at the cavity is probed at rate
$\delta$. When a probe event occurs, the server informs the dispatcher about the actual queue length
and the dispatcher updates its estimate accordingly. This may seem to imply that a jump
occurs from state $(q,e,j)$ to state $(q,q,j)$. However, when $q < m$ the new estimated queue
length is below $m$ and in the large-scale limit this implies that the dispatcher instantaneously
assigns a batch of jobs such that the actual queue length becomes $m$. Hence, at rate $\delta$,
probe events cause a state change from $(q,e,j)$ to $(\max(q,m),\max(q,m),j)$. Likewise
when the state is $(0,e)$ a jump occurs to state $(m,m,j')$ at rate $\delta \alpha_{j'}$.

Finally, state changes also occur at some unknown rate $\nu$ when the dispatcher assigns a new incoming job to the queue at the cavity which has an estimated queue length equal to $m$. Let $\pi^{e\shortrightarrow}_{e^\prime}(m,\nu)$ denote the probability that the estimated queue length equals $e^\prime$ for $e^\prime=m,m+1$ 
and let $\pi^{a\shortrightarrow}_q(m,\nu)$ denote the probability that the actual queue length equals $q$ for $q=0,...,m+1$. At first glance it may appear that the rate $\nu$
at which the dispatcher changes the state due to new arrivals is such that $\nu \pi^{e\shortrightarrow}_m(m,\nu)$ should equal $\lambda$, as new arrivals are assigned randomly to a server with the lowest estimated queue length. However, keep in mind that part of the arrival rate is already consumed by the batch assignments that accompanied the probe events. The rate consumed
by these batch arrivals equals $\delta \sum_{q=0}^m (m-q) \pi^{a\shortrightarrow}_q(m,\nu)$ as $m-q$ jobs
are instantaneously assigned when a probe event reveals a server with queue length $q \leq m$.
The rate $\nu$ should therefore obey
\begin{align}\label{eq:nu_check_push} \nu \pi^{e\shortrightarrow}_m(m,\nu) = \lambda - \delta \sum_{q=0}^m (m-q) \pi^{a\shortrightarrow}_q(m,\nu).
\end{align}
We are now in a position to define the rate matrix $Q^{\shortrightarrow}(m,\nu)$ of the queue at the cavity on
the state space $\Omega^\shortrightarrow$:
\begin{align}\label{Qpush}
Q^{\shortrightarrow}(m,\nu) = \begin{bmatrix}
Q^{\shortrightarrow}_{0,0}(\nu) & Q^{\shortrightarrow}_{0,1}(\nu) & &  & Q^{\shortrightarrow}_{0,m} & \\
Q_{1,0}^{\shortrightarrow} & Q_{1,1}^{\shortrightarrow}(\nu) & Q_{1,2}^{\shortrightarrow}(\nu)  & & Q^{\shortrightarrow}_{1,m} & \\
 & Q^{\shortrightarrow}_{2,1} & Q^{\shortrightarrow}_{2,2}(\nu) & \ddots   & Q^{\shortrightarrow}_{2,m} & \\
&  &\ddots& \ddots &  \vdots &  \\
&  &&   Q^{\shortrightarrow}_{m,m-1}&  Q^{\shortrightarrow}_{m,m}(\nu) & Q^{\shortrightarrow}_{m,m+1}(\nu)  \\
&  &&  &  Q^{\shortrightarrow}_{m+1,m}& Q^{\shortrightarrow}_{m+1,m+1}(\nu)
\end{bmatrix},
\end{align}
where the matrix $Q_{q,q'}^{\shortrightarrow}(\nu)$ captures the changes from states with actual queue length $q$ to states with actual queue length $q'$.

We now define the matrices $Q_{q,q'}^{\shortrightarrow}(\nu)$ for all possible combinations of $q$ and $q'$.
Due to our discussion on the service completions we have 
\[Q^{\shortrightarrow}_{q,q-1} = \begin{bmatrix}
	s^*\alpha & 0 \\  0 & s^*\alpha
\end{bmatrix}, \ \ 
Q^{\shortrightarrow}_{1,0} = \begin{bmatrix}
	s^* & 0 \\ 0 & s^*
\end{bmatrix}
\mbox{  \ \ and \ \  }   
Q^{\shortrightarrow}_{m+1,m} = \begin{bmatrix}
	0 & s^*\alpha
\end{bmatrix},\] 
for $q =2,\ldots,m$.  
The block diagonal structure indicates that the estimated queue length $e \in \{m,m+1\}$
is not updated when a service completion occurs. Further note that when the
actual queue length $q=m+1$, then $e=m+1$ as well as $e \geq q$. The probe events that occur when
the queue length $q$ is below $m-1$ immediately increase $q$ to $m$, therefore
\[Q^{\shortrightarrow}_{q,m} = \begin{bmatrix}
	\delta I & 0 \\ \delta I & 0
\end{bmatrix} 
\mbox{ \ \ and \ \   }   
Q^{\shortrightarrow}_{0,m} = \begin{bmatrix}
	\delta \alpha & 0 \\ \delta \alpha & 0
\end{bmatrix},\] 
with $I$ the $n_s \times n_s$ identity matrix and $0 < q < m-1$. Note that $e=m$ after such a probe event (as the second block column is zero).
The job assignments at rate $\nu$ increase the queue length by one and can only occur when
the estimated queue length is $m$, hence for $q=1,\ldots,m-2,m$
\[Q^{\shortrightarrow}_{q,q+1}(\nu) = \begin{bmatrix}
	0 & \nu I \\ 0 & 0
\end{bmatrix}, \ \ 
Q^{\shortrightarrow}_{m-1,m}(\nu) = \begin{bmatrix}
	\delta I & \nu I \\ \delta I & 0
\end{bmatrix} 
\mbox{ \ \ and \ \  }   
Q^{\shortrightarrow}_{0,1}(\nu) = \begin{bmatrix}
	0 & \nu \alpha \\ 0 & 0
\end{bmatrix},\] 
where we note that the estimated queue length becomes $m+1$.
The matrix $Q^{\shortrightarrow}_{m-1,m}(\nu)$ captures both job assignments at rate $\nu$
when $e=m$ and probe events at rate $\delta$.
Finally the diagonal blocks capture changes in the service phase, therefore we have
\[Q^{\shortrightarrow}_{q,q}(\nu) = \begin{bmatrix}
	S-(\nu+\delta)I & 0 \\ 0 & S-\delta I
\end{bmatrix}, \ \
Q^{\shortrightarrow}_{m,m}(\nu) = \begin{bmatrix}
	S-\nu I & 0 \\ \delta I & S-\delta I
\end{bmatrix}
\mbox{  \ \ and \ \    }   
Q^{\shortrightarrow}_{0,0}(\nu) = \begin{bmatrix}
	-(\nu+\delta) & 0 \\ 0 & -\delta 
\end{bmatrix},\] 
for $q=1,\ldots,m-1$ and $Q^{\shortrightarrow}_{m+1,m+1} = S$, where the $\delta I$ in $Q^{\shortrightarrow}_{m,m}$
is due to the fact that $e$ is updated to $m$
if a probe arrives when the state is of the form $(m,m+1,j)$. 

Note that both $m$ and $\nu$ are unknown at this stage and we 
indicate how to determine both next.

\subsection{Finding $m$ and $\nu$}
To assess the performance in the large-scale limit we first need to determine the
unknowns $m$ and $\nu$. It is not hard to see that the probability that the queue
is empty, denoted as $\pi^{a\shortrightarrow}_0(m,\nu)$, decreases as $\nu$ increases. Indeed, if we number the
states lexicographically, the transitions with rate $\nu$ increase the state and
\cite{busicMOR}[Theorem 1] implies that the probability to be in the first $b$ states, for any $b$, decreases as $\nu$ increases. As the first two states correspond to an empty queue, setting $b=2$ yields the result.

Another observation is that if we set $\nu=0$, then all the states with $e=m+1$ are transient,
meaning all queues have an estimated queue length equal to $m$, that is, $\pi^{e\shortrightarrow}_m(m,\nu) = 1$ and
$\pi^{e\shortrightarrow}_{m+1}(m,\nu) = 0$. Further, setting $\nu=\infty$ implies
that the states with $e=m$ are transient and all queues have estimated queue length $m+1$,
that is, $\pi^{e\shortrightarrow}_m(m,\nu) = 0$ and
$\pi^{e\shortrightarrow}_{m+1}(m,\nu) = 1$.

Combining these two observations, we note that $\pi^{a\shortrightarrow}_0(m,\nu) > \pi^{a\shortrightarrow}_0(m',\nu')$
if $m < m'$ or $m=m'$ and $\nu < \nu'$. This implies that there exists a unique $(m,\nu)$ such that $\pi^{a\shortrightarrow}_0(m,\nu) = 1-\lambda$. We now derive an explicit expression for $m$ by
studying the Markov chain with $\nu=0$ characterized by $Q^{\shortrightarrow}(m,0)$. If we remove the transient states with $e=m+1$, this chain evolves on the state space
\[ \Omega_{(0)}^\shortrightarrow = \{0\} \cup \{(q,j) |  q = 1,\ldots,m; j = 1,\ldots,n_s \},\]
and has rate matrix $Q^{\shortrightarrow}_{(0)}(m)$ given by
\begin{align}\label{eq:Qpush0}
Q^{\shortrightarrow}_{(0)}(m) = \begin{bmatrix}
-\delta  & & & & & \delta \alpha \\
s^* & S-\delta I & &  & & \delta I\\
 & s^*\alpha & S-\delta I &  & & \delta I\\
&  &\ddots& \ddots & & \vdots   \\
& & & s^*\alpha & S-\delta I &\delta I\\
&  &&  & s^* \alpha &  S 
\end{bmatrix}.
\end{align}
Let $\pi^{a\shortrightarrow}_0(m)$ be the steady state probability that the cavity queue is empty and $\pi^{a\shortrightarrow}_q(m)$ the steady state probability that we are in a state
of the form $(q,j)$
(note that these are the same as $\pi^{a\shortrightarrow}_q(m,0)$ defined before).

\begin{proposition}\label{th:pim}
The steady state probabilities of $Q^{\shortrightarrow}_{(0)}(m)$ are such that for $i=1,\ldots,m$
\begin{align}\label{pim}
\sum_{q=0}^{i-1} \pi^{a\shortrightarrow}_q(m) = \frac{y^{m-i}/\delta}{1/\delta + y^{m-1} + (1-y^{m-1}) \alpha'(-S)^{-1}\mathbf{1}},
\end{align}
with $y$ as in \eqref{eq:y} and $\alpha' = \alpha (\delta I - S)^{-1}/\alpha (\delta I - S)^{-1}\mathbf{1}$. 
\end{proposition}
\begin{proof} 
	Let $S_{<i} \subset \Omega_{(0)}^\shortrightarrow$ be the set of states with $q < i$. 
	We refer to the set of states of the form $(m,j)$ as level $m$ of the chain.
	We divide time into cycles that start whenever the chain leaves level $m$.
	Note that as these points in time are renewal points, 
	the sum $\sum_{q=0}^{i-1} \pi^{a\shortrightarrow}_q(m)$ can be expressed as the mean time
	the chain spends in the set $S_{<i}$ during a single cycle, divided by the mean cycle length.
	The mean cycle length is given by the mean time away from level $m$ plus the mean time in level $m$. 
	
	The time that the chain is away from level $m$ has an exponential distribution with
	parameter $\delta$, so the mean time away is $1/\delta$. It therefore suffices to
	argue that $ y^{m-1} + (1-y^{m-1}) \alpha'(-S)^{-1}\textbf{1}$ is the mean time spend in level $m$
	in order to show that $1/\delta + y^{m-1} + (1-y^{m-1}) \alpha'(-S)^{-1}\textbf{1}$ is the mean cycle
	length.
	
	When the service of a job starts, it completes before an exponential
	timer with parameter $\delta$ expires with probability 
	\[y=\int_0^\infty \alpha e^{St} s^* e^{-\delta t} dt = \alpha (\delta I - S)^{-1} s^*.\]
	Note that $y$ can also be expressed as
	\begin{align}
		y &= \int_0^{+\infty} P[Y < X|X=t]f_X(t)dt = \int_0^{+\infty} (1-\alpha e^{St}\textbf{1})\delta e^{-\delta t}dt = 1- \alpha(\delta I -S)^{-1}\textbf{1}\delta,
		\label{eq:ydef2}
	\end{align}
	where $X$ is exponential with parameter $\delta$ and $Y$ has a PH distribution with parameters $(\alpha,S)$.
	
	Given that the exponential timer expires first, the distribution of the service phase when the timer expires is given by
	\[\alpha'=\int_0^\infty \alpha e^{St} \delta e^{-\delta t} dt/(1-y) = 
	\alpha (\delta I - S)^{-1} \delta / \alpha (\delta I - S)^{-1}\textbf{1} \delta,\]
	where we have used \eqref{eq:ydef2} in the second equality. 
	The expression for the mean cycle length therefore follows by noting  that with probability $y^{m-1}$ the queue was empty
	just prior to entering level $m$ and therefore the time spend in level $m$
	equals the mean service time, which is $1$. While with probability $1-y^{m-1}$,
	the queue did not become empty because an exponential timer with parameter $\delta$ 
	expired during the service of a job 
	and this implies that the time in level $m$ has a PH distribution with
	parameters $(\alpha',S)$, which has a mean given by $\alpha' (-S)^{-1}\textbf{1}$.  
	
	Finally, in order to express the mean time in the set $S_{<i}$ during a cycle, we note that
	$y^{m-i}$ is the probability that the set $S_{<i}$ is visited once, while with probability
	$1-y^{m-i}$ the set $S_{<i}$ is not visited during a cycle. The mean sojourn time in 
	the set $S_{<i}$ is clearly $1/\delta$, which yields that $y^{m-i}/\delta$ is the mean time
	in the set $S_{<i}$ during a cycle. 
\end{proof}
The above result can also be deduced in an algebraic manner based on the balance equations.

\begin{lemma} \label{lem:Z-X}
Let $X$ be exponential with parameter $\delta$ and $Z$ a distribution on $\mathbb{R}^+$,
then
\begin{align}\label{eq:lemma}
1 +\delta  E[Z-X|Z>X] = \frac{\delta E[Z]}{P[Z>X]}.
\end{align}
\end{lemma}
\begin{proof}
We need to argue that $E[Z] = P[Z>X]E[X]+ E[Z-X|Z>X] P[Z>X]$ as $E[X]=1/\delta$.
We have $E[Z] = E[\min(Z,X)] + P[Z> X] E[Z-X|Z>X]$ and the result follows provided
that $E[\min(Z,X)] = P[Z>X]E[X]$. This holds for general $Z$ and $X$ exponential
as $E[X] = E[\min(Z,X)] + P[X>Z] E[X-Z|X>Z] =  E[\min(Z,X)] + P[X>Z] E[X]$.
\end{proof}

\begin{theorem}\label{cor:push}
For the push policy with arrival rate $\lambda$, probe rate $\delta$ and $y$ as in \eqref{eq:y}, we have $1 -  \lambda \in [\pi^{a\shortrightarrow}_0(\lceil \tilde m \rceil),
\pi^{a\shortrightarrow}_0(\lfloor \tilde m \rfloor)]$ for
\begin{align}\label{eq:push_m}
\tilde m = \frac{\log\left( \frac{1}{y} + \left( \frac{\lambda}{\delta (1-\lambda)} - 1 \right) \cdot \frac{1-y}{y} \right)}{\log(1/y)},
\end{align}
meaning $\lceil \tilde m \rceil$
is the maximum queue length for the queue at the cavity.
\end{theorem}
\begin{proof}
When $\pi^{a\shortrightarrow}_0(m)=1-z \in (0,1)$, we have by Proposition \ref{th:pim} with $i=1$
\begin{align}\label{eq:push_lam}
z = 1 - \frac{1}{\delta-\delta \alpha'(-S)^{-1}\textbf{1}+(1+\delta \alpha'(-S)^{-1}\textbf{1})/y^{m-1}}.
\end{align}
Thus $z$ increases as a function of $m$ and  $z = \lambda$ if $m=\tilde m$ with
 \[ \tilde m = 1 - \left.\log\left(1+\frac{\lambda/(1- \lambda) - \delta  }
{1+\delta \alpha' (-S)^{-1}\textbf{1}}\right)\middle/\log(y)\right..\]
The result now follows from \eqref{eq:lemma} as $\alpha' (-S)^{-1}\textbf{1}=E[Z-X|Z>X]$, $P[Z<X]=y$
and $E[Z]=1$.
\end{proof}

When the job sizes are exponential with mean $1$, we have $P[Z<X]=1/(1+\delta)$.
This implies that $\tilde m=-\log(1- \lambda)/\log(1+\delta)$
and $\lambda = 1 - (1+\delta)^{-\tilde m}$, which are the expressions 
derived in \cite{van2019hyper} for the fixed point of a set of drift equations. Furthermore, it is easy to see that we still have vanishing waiting times whenever $\frac{\lambda}{(1-\lambda)} \leq \delta$ irrespective of the job size distribution.

\begin{remark}
	The proofs of Proposition \ref{th:pim} and Theorem \ref{cor:push} can easily be generalized to include any positive valued distribution. One finds that the value for $\tilde m$ obtained in \eqref{eq:push_m} holds for any job size distribution $Z$ and this allows one to generalize the results of Corollary \ref{th:lam_m}, Theorems \ref{th:mbounds}, \ref{th:qact} and Corrolary \ref{th:push_critical} for general job sizes with $y = P[Z < X]$.
\end{remark}

\begin{example}
When the job sizes follow an Erlang-$k$ distribution with mean $1$, we have
$y=(k/(k+\delta))^k$ and therefore $\tilde m = m_{Erl}(k)$ with
\begin{align}\label{eq:mErl}
m_{Erl}(k) = 1 - \left.\log\left(
1+ \frac{1}{\delta} \left(  \frac{\lambda}{1-\lambda} - \delta \right) 
\left( 1- \left( \frac{k}{k+\delta}\right)^k \right)\right)\middle/k \log
\left(\frac{k}{k+\delta}\right)\right.,
\end{align}
and 
\begin{align}\label{eq:det}
\lim_{k \shortrightarrow \infty} m_{Erl}(k) = \frac{1}{\delta} \log \left( 
1 +  \frac{1}{\delta}   \frac{ \lambda}{1-\lambda} (e^\delta-1)\right),
\end{align}
as $\lim_{k \shortrightarrow \infty} (k/(k+\delta))^k = e^{-\delta}$.
\end{example}
The expression for $\tilde m$ presented in \eqref{eq:push_m} is in general not an integer.
In order to find the proper $(m,\nu)$ pair for the queue at the cavity, we propose the following algorithm:
\begin{itemize}
	\item Set $m = \lfloor \tilde m \rfloor$, with $\tilde m$ as defined in \eqref{eq:push_m}.
	\item Determine the unique rate $\nu \geq 0$ such that $\pi^{a\shortrightarrow}_0( \lfloor \tilde m \rfloor,\nu)=1-\lambda$ using a bisection algorithm by repeatedly computing the stationary distribution of \eqref{Qpush}.
\end{itemize}
Given a probe rate $\delta> 0$ and a job size distribution, we can find the $\lambda$ values at which $\tilde m$ takes integer values (and vice versa we find the $\delta$ values given a fixed $\lambda$).

\begin{corollary}\label{th:lam_m}
In the same setting as Theorem \ref{cor:push}, we find that the maximum queue length of the queue at the cavity is equal to $m > 0$ for $\lambda \in (\lambda^{\shortrightarrow}_{m-1},\lambda^{\shortrightarrow}_m]$
with 
\begin{align}\label{eq:lam_m}
\lambda^{\shortrightarrow}_m = \frac{\delta y (1-y^m)}{\delta y (1-y^m)+y^m(1-y)}.
\end{align}
Further, the maximum queue length of the queue at the cavity is equal to $m > 0$ 
for $\delta \in [\delta^{\shortrightarrow}_{m},\delta^{\shortrightarrow}_{m-1})$
with
\begin{align}\label{eq:delta_m}
\delta^{\shortrightarrow}_m = \frac{y^{m-1} (1-y)}{1-y^m}
\frac{\lambda}{1-\lambda}.
\end{align}
\end{corollary} 
\begin{proof}
If $\pi^{a\shortrightarrow}_0(m)=1-z \in (0,1)$, then \eqref{eq:push_lam} holds.
As \eqref{eq:lemma} corresponds to stating that $1+\delta\alpha' (-S)^{-1} \textbf{1} = \delta/(1-y)$,
we can use this equality twice in \eqref{eq:push_lam} to find that $z=\lambda^{\shortrightarrow}_m$. Hence, $1-\pi^{a\shortrightarrow}_0(m)=\lambda^{\shortrightarrow}_m$ and the maximum queue length increases by
one whenever $\lambda$ is such that $\pi^{a\shortrightarrow}_0(m) = 1-\lambda$ 
for some integer $m$. \eqref{eq:delta_m} is immediate from \eqref{eq:lam_m}
by setting $\lambda^{\shortrightarrow}_m =\lambda$, $\delta=\delta^{\shortrightarrow}_m$
and solving for $\delta^{\shortrightarrow}_m$.
\end{proof}
For exponential job sizes $y=1/(1+\delta)$ and $\lambda^{\shortrightarrow}_m$ simplifies to 
$1-(1+\delta)^{-m}$ and $\delta^{\shortrightarrow}_m$ becomes
$(1-\lambda)^{-1/m}-1$. 
For completeness we end this subsection by showing that 
\eqref{eq:nu_check_push} holds. 

\begin{proposition}\label{th:nu_eq_holds}
In the same setting as Theorem \ref{cor:push}, equation \eqref{eq:nu_check_push} holds for the steady state probabilities of the
Markov chain characterized by \eqref{Qpush} when
$\pi_0^{a\shortrightarrow}(m,\nu)=1-\lambda$.
\end{proposition}
\begin{proof}

Let $\pi^{\shortrightarrow}_{q,m}(m,\nu)$ be the probability that the
actual queue length equals $q$ and the estimated queue length equals $m$.
Let $\pi^{\shortrightarrow}_{q,e,j}(m,\nu)$ be the probability that the chain
characterized by \eqref{Qpush} is in state $(q,e,j)$.
The rate of making a jump from an actual queue length below $q$ to an actual queue
length of at least $q$, for $q=1,\ldots,m$, is given by
\[\delta \sum_{j=0}^{q-1} \pi_j^{a\shortrightarrow}(m,\nu) +
\pi^{\shortrightarrow}_{q-1,m}(m,\nu) \nu, \]
and this rate equals the rate of making a jump from an actual queue length
of $q$ to $q-1$ (as the queue length can only decrease by one), which is given by 
\[\sum_{e,j} \pi^{\shortrightarrow}_{q,e,j}(m,\nu) s_j^*.\]
Summing this equality for $q=1$ to $m$ yields
\begin{align}\label{eq:sumb}
	\delta \sum_{q=0}^{m-1} &(m-q) \pi^{a\shortrightarrow}_q(m,\nu) + 
	\nu (\pi^{e\shortrightarrow}_m(m,\nu)-\pi^{\shortrightarrow}_{m,m}(m,\nu)) = \sum_{q=1}^m \sum_{e,j} \pi^{\shortrightarrow}_{q,e,j}(m,\nu) s_j^*. 
\end{align} 
The rate of jumping from an actual queue length of $m$ to $m+1$ is given by
$\pi^{\shortrightarrow}_{m,m}(m,\nu) \nu$ and this rate equals the jump rate
from an actual queue length of $m+1$ to $m$ given by $\sum_j \pi^{\shortrightarrow}_{m+1,m+1,j}(m,\nu) s_j^*$. If we combine this equality with \eqref{eq:sumb}, we find that
\begin{align}\label{eq:sumb2}
	\delta \sum_{q=0}^{m-1} &(m-q) \pi^{a\shortrightarrow}_q(m,\nu) + 
	\nu \pi^{e\shortrightarrow}_m(m,\nu) = \sum_{q=1}^{m+1} \sum_{e,j} \pi^{\shortrightarrow}_{q,e,j}(m,\nu) s_j^*. 
\end{align} 
Equation \eqref{eq:nu_check_push} then follows provided that the right-hand side of the
above equality equals $\lambda$.

If we observe the queue when it is busy and focus only on the phase process,
we obtain a Markov chain with rate matrix $S+s^*\alpha$, therefore the right-hand side
of \eqref{eq:sumb2} equals $(1- \pi_0^{a\shortrightarrow}(m,\nu)) \beta s^*$, where
$\beta$ is the unique invariant vector of  $S+s^*\alpha$. As $\beta s^*$ is equal to the
mean service time of a job, which equals $1$, the right-hand side of \eqref{eq:sumb2} 
becomes $\lambda$ when $\pi_0^{a\shortrightarrow}(m,\nu)=1-\lambda$.
\end{proof}

\subsection{Performance bounds}
The result presented in Theorem \ref{cor:push} implies that the actual queue length of the
queue at the cavity is bounded by $\lceil \tilde m \rceil$ and this bound is sensitive
to the phase-type job size distribution characterized by $(\alpha,S)$
via the probability $y=P[Z<X]$. 
We now present tight upper and lower bounds on  $\lceil \tilde m \rceil$.

\begin{theorem}\label{th:mbounds}
In the same setting as Theorem \ref{cor:push}, the maximum  queue length $\lceil \tilde m \rceil$ for the queue at the cavity is  such that
\begin{align}\label{eq:mbounds}
\left\lceil \frac{1}{\delta} \log \left( 
1 +  \frac{1}{\delta}   \frac{ \lambda}{1-\lambda} (e^\delta-1)\right) \right\rceil
\leq 
\lceil \tilde m \rceil \leq 
\left\lceil \frac{\lambda}{(1-\lambda)\delta} \right\rceil, 
\end{align}
for any PH distribution and these bounds are tight.
\end{theorem}
\begin{proof}
We start by noting that any job size distribution $Z$ that maximizes 
$\lambda^{\shortrightarrow}_m$ for all $m$ minimizes $\tilde m$ and likewise any $Z$
that minimizes $\lambda^{\shortrightarrow}_m$ for all $m$ maximizes $\tilde m$.
We now show that $\lambda^{\shortrightarrow}_m$ decreases as a function of $y$
for $y \in (0,1)$, which is equivalent to showing that $\kappa(y) =y^{m-1}(1-y)/ (1-y^m)$ increases
in $y$. One readily
checks that $\kappa'(y) > 0$ if
$m(1-y)-(1-y^m)$ is positive, which holds for $y \in (0,1)$ and $m \geq 1$.
In other words,  $\tilde m$ is minimized/maximized by the distribution $Z$ that
minimizes/maximizes $y = P[Z < X]$. 

When $Z$ is deterministic we have $P[Z<X] = e^{-\delta}$ and
by Jensen's inequality we have for any $Z$ with $E[Z]=1$ that
\[e^{-\delta} = e^{-\delta E[Z]} \leq E[e^{-\delta Z} ]= \int_0^\infty e^{-\delta s} dP[Z\leq s] =
P[Z<X],\]
which implies that $y\geq e^{-\delta}$. 
Plugging $y=e^{-\delta}$ in \eqref{eq:push_m} and using
$1=\log(e^{\delta})/\delta$ yields the lower bound and its tightness follows from  \eqref{eq:det}.

To prove the upper bound and the fact that it is tight, 
consider the order $2$ hyperexponential distribution $Z(\varepsilon)$ 
with $p_1 = 1-\varepsilon$, $p_2=\varepsilon$, $\mu_1=(1-\varepsilon)/\varepsilon$ and
$\mu_2 = \varepsilon/(1-\varepsilon)$. We have $E[Z(\varepsilon)] = p_1/\mu_1+p_2/\mu_2 = 1$ and
\[P[Z(\varepsilon) < X] = \sum_{i=1}^2 p_i \frac{\mu_i}{\mu_i+\delta}=\frac{(1-\varepsilon)^2}{(1-\varepsilon)+\delta \varepsilon}+\frac{\varepsilon^2}{\varepsilon+(1-\varepsilon)\delta} ,\]
meaning $y=P[Z(\varepsilon)<X]$ tends to one as $\varepsilon$ tends to zero.
Using \eqref{eq:push_m} and the continuity and Taylor series expansion
of $\log(1+z) = \sum_{i=1}^\infty (-1)^{i+1} z^i/i$ in $z=0$, this yields
\[ \lim_{\varepsilon \shortrightarrow 0} \tilde m  = 
1 + \lim_{\varepsilon \shortrightarrow 0} \left[ \left.
\left(\frac{\lambda}{1- \lambda} - \delta \right)
\frac{P[Z(\varepsilon)> X]}{\delta } \middle/ (1-P[Z(\varepsilon)<X]) \right. \right]
= \frac{\lambda}{(1- \lambda)  \delta}. \]
\end{proof}

We can make the following observations.
\begin{enumerate}
\item The upper bound becomes $1$ when $\delta \geq \lambda/(1-\lambda)$.
Hence, for  $\delta \geq \lambda/(1-\lambda)$ we have vanishing wait for any 
PH job size
distribution.
\item While the upper bound on the maximum queue length grows as $1/(1-\lambda)$ for fixed 
$\delta$ when $\lambda$ tends to one, it should be noted that for any given PH
distribution we have $y < 1$, which implies that for a given
PH distribution  the
maximum queue length only grows as fast as $\log(1/(1-\lambda))$ 
when  $\lambda$ tends to one (as in the exponential case).
As such, the limits of $\lambda$ and $y$ tending to one cannot be
interchanged.  
\item The upper bound can also be established based on \eqref{eq:nu_check_push}, by
noting that
\begin{align*}
 0 \leq \nu \pi^{e\shortrightarrow}_m(m,\nu) &= \lambda - \delta \sum_{q=0}^m (m-q) \pi^{a\shortrightarrow}_q(m,\nu)\leq \lambda-\delta m (1-\lambda),
\end{align*}
as $\pi^{a\shortrightarrow}_q(m,\nu)\geq 0$ for $q > 0$ and $\pi_0^{a\shortrightarrow}(m,\nu)=1-\lambda$.
Hence, $m \leq \lambda/((1-\lambda)\delta)$ as required.
\end{enumerate}
 
From the proof of Theorem \ref{th:mbounds} we observe that the upper bound is even tight
for the class of order $2$ phase-type distributions. The lower bound however is not tight
if we restrict ourselves to order $k$ phase-type distributions. The next
result shows that for order $k$ phase-type distributions the lower bound corresponds
to Erlang-$k$ service times:

\begin{proposition}\label{th:mbounds_orderk}
In the same setting as Theorem \ref{cor:push}, the maximum  queue length $\lceil \tilde m \rceil$ for the queue at the cavity is  such that
$\left\lceil m_{Erl}(k) \right\rceil
\leq 
\lceil \tilde m \rceil$, 
for any order $k$ phase-type distribution, where $m_{Erl}(k)$ was defined in \eqref{eq:mErl}.
\end{proposition}
\begin{proof}
	Let $Z$ be a random variable with an order $k$ representation. From the
	proof of Theorem \ref{th:mbounds} it suffices to show that $P[Z < X]$, with $X$
	exponential with parameter $\delta$, is larger than the probability $P[Z' < X]$, where
	$Z'$ is an Erlang-$k$ random variable. As $P[Z < X] = E[e^{-\delta Z}]$,
	it suffices to show that $E[\xi(Z)] \geq E[\xi(Z')]$ holds for any convex
	function $\xi$. Theorem 3 in \cite{ocinneide2} shows that any PH distribution
	with an order $k$ representation {\it majorizes} the order $k$ Erlang distribution with
	the same mean, where a distribution $Z_1$ majorizes another distribution $Z_2$ exactly when $E[\xi(Z_1)] \geq E[\xi(Z_2)]$ for any convex function $\xi$ \cite{hardy1952}.
\end{proof}
Some remarks:
\begin{enumerate}
\item 
It is easy to check that $m_{Erl}(k)$ is decreasing in $k$, which is in agreement with
the fact that any PH distribution with an order $k$ representation also has an
order $k'$ representation for any $k' > k$. 
\item $Z_2$ majorizes $Z_1$ 
if $Z_1 \leq_{cx} Z_2$, where $\leq_{cx}$ is the usual convex ordering. This
is also equivalent to stating that $E[Z_1]=E[Z_2]$ and $E[\max(Z_1 - t,0)] \leq E[\max(Z_2 - t,0)]$ for any $t$. This allows us to show that if $Z_1 \leq_{cx} Z_2$ for two
job size distributions $Z_1$ and $Z_2$ (both with mean $1$), the maximum queue length for $Z_2$ is lower bounded by the maximum queue length for $Z_1$.
\end{enumerate} 
We proceed by presenting an explicit lower and upper bound on the mean queue length of the queue at the cavity:
 
\begin{theorem}\label{th:qact}
In the same setting as Theorem \ref{cor:push}, let $E[Q^{a\shortrightarrow}]$ be the mean queue length of the queue at the cavity, then

\[\lfloor \tilde m \rfloor-\lambda^{\shortrightarrow}_{\lfloor \tilde m \rfloor}/\delta \leq E[Q^{a\shortrightarrow}] \leq 
\lceil \tilde m \rceil-\lambda^{\shortrightarrow}_{\lceil \tilde m \rceil}/\delta,\] 
where
$\tilde m$ is given by \eqref{eq:push_m} and $\lambda^{\shortrightarrow}_m$ is given by \eqref{eq:lam_m}.
\end{theorem} 
\begin{proof}
Let $(\lfloor \tilde m \rfloor,\nu)$ be such that $\pi^{a\shortrightarrow}_0(\lfloor \tilde m \rfloor,\nu)=1-\lambda$ (with $\nu=0$ if $\tilde m$ is integer) and let
\[E[Q^{a\shortrightarrow}] = \sum_{i=1}^{\lceil \tilde m \rceil}
i \pi^{a\shortrightarrow}_i(\lfloor \tilde m \rfloor,\nu)= \sum_{i=1}^{\lceil \tilde m \rceil}
\sum_{j=i}^{\lceil \tilde m \rceil} \pi^{a\shortrightarrow}_j(\lfloor \tilde m \rfloor,\nu),\] 
denote the mean queue length. Due to \cite{busicMOR}, the probability to have an actual queue length
of at least $i$ grows with $\nu$, which implies 
\[ 
\sum_{i=1}^{\lfloor \tilde m \rfloor}\sum_{j=i}^{\lfloor \tilde m \rfloor} \pi^{a\shortrightarrow}_j(\lfloor \tilde m \rfloor,0) \leq
E[Q^{a\shortrightarrow}] \leq 
\sum_{i=1}^{\lceil \tilde m \rceil}\sum_{j=i}^{\lceil \tilde m \rceil} \pi^{a\shortrightarrow}_j(\lceil \tilde m \rceil,0).
\]
Hence, it suffices to derive an expression for 
\[
\sum_{i=1}^m \sum_{j=i}^m \pi^{a\shortrightarrow}_j(m,0)
= \sum_{i=1}^m \left( 1- \sum_{j=0}^{i-1} \pi^{a\shortrightarrow}_j(m)\right),
\] 
with $m>0$ an integer.
Proposition \ref{th:pim} yields
\begin{align*}
	\sum_{i=1}^m \left( 1- \sum_{j=0}^{i-1} \pi^{a\shortrightarrow}_j(m)\right)=m - \frac{(1-y^m)/(1-y)}{1+\delta y^{m-1} + \delta(1-y^{m-1})\alpha'(-S)^{-1}\textbf{1}}.
\end{align*}
Using \eqref{eq:lemma} twice, we find
\[
\sum_{i=1}^m \left( 1- \sum_{j=0}^{i-1} \pi^{a\shortrightarrow}_j(m)\right)=
m - \frac{(1-y^m)}{\delta(1-y^m)+y^{m-1}(1-y)} = m - \lambda^{\shortrightarrow}_m/\delta,
\]
which completes the proof.
\end{proof} 

The following observations are worth noting:
\begin{enumerate}
\item The difference between the upper and lower bound is less than $1$ as
$\lambda^{\shortrightarrow}_m$ is increasing in $m$. 
\item When the job size is exponential, the lower and upper bound are given by
$\lfloor \tilde m \rfloor - (1-(1+\delta)^{-\lfloor \tilde m \rfloor})/\delta$ 
and $\lceil \tilde m \rceil - (1-(1+\delta)^{-\lceil \tilde m \rceil})/\delta$, respectively,
as $\lambda^{\shortrightarrow}_m = 1-(1+\delta)^{-m}$, which are the bounds presented in \cite{van2019hyper}.
\item The expression for the upper and lower bound can also be derived using \eqref{eq:nu_check_push}. More specifically, when $\lambda=\lambda_m^{\shortrightarrow}$,
then $\nu =0$ and \eqref{eq:nu_check_push} yields that
\[\lambda = \delta \sum_{q=0}^{m} (m-q) P[Q^{a\shortrightarrow} = q]
= \delta m - \delta E[Q^{a\shortrightarrow}].\]
\item In order to make the bounds given in Theorem \ref{th:qact} independent of the job size distribution, one can note that $\lambda_m^\shortrightarrow$ (given by \eqref{eq:lam_m}) is decreasing in function of $y$ (for any $m \geq 1$). One can therefore take the limit $y \rightarrow 1^-$ and $y \rightarrow \left( e^{-\delta} \right)^+$ to obtain an upper and a lower bound. For the upper bound, we obtain the expression:
$$
E[Q^{a\shortrightarrow}] \leq 
\left. \delta \cdot \left \lceil \frac{\lambda}{(1-\lambda) \delta} \right \rceil ^2
\middle/
\left( 1+\delta \cdot \left \lceil \frac{\lambda}{(1-\lambda) \delta} \right \rceil \right) \right.
$$
\end{enumerate}

\subsection{Critically loaded system}
We now proceed our analysis of the push policy by considering the limiting regime $\lambda \shortrightarrow 1^-$ (see e.g.~\cite{hellemans2021mean, maguluri2016heavy, mitzenmacher2}). 
\begin{corollary} \label{th:push_critical}
	In the same setting as Theorem \ref{cor:push}, let $E[Q^{a\shortrightarrow}]$ resp.~$E[R^{a\shortrightarrow}]$ denote the mean queue length resp.~mean response time of the pull policy with arrival rate $\lambda$, then:
	\begin{equation}\label{eq:heavy_pull}
		\lim_{\lambda \shortrightarrow 1^-} \frac{E[R^{a \shortrightarrow}]}{\log\left(\frac{1}{1-\lambda}\right)}=
		\lim_{\lambda \shortrightarrow 1^-} \frac{E[Q^{a \shortrightarrow}]}{\log\left(\frac{1}{1-\lambda}\right)}=
		\lim_{\lambda \shortrightarrow 1^-} \frac{\tilde m}{\log\left(\frac{1}{1-\lambda}\right)}
		= \frac{1}{\log(1/y)}.
	\end{equation}
\end{corollary}
\begin{proof}
	The first equality follows from Little's law, while the second equality easily follows by applying Theorem \ref{th:mbounds}. The last equality follows from \eqref{eq:push_m} after computing (using l'H\^opital's rule):
	$$
	\lim_{\lambda \shortrightarrow 1^-} \left. \log\left( 1 + \left( \frac{\lambda}{1-\lambda} - \delta \right) \frac{1-y}{\delta} \right) \middle/
	\log\left(\frac{1}{1-\lambda}\right)\right.  = 1 .
	$$
\end{proof}
	As $y\shortrightarrow 1^-$ we find that the right hand side in \eqref{eq:heavy_pull} goes to infinity. This corresponds to the fact that we found a (tight) upper bound in Theorem \ref{th:mbounds} which is of the order $\frac{1}{1-\lambda}$ rather than $\log\left( \frac{1}{1-\lambda} \right)$. In the next section, we perform some numerical experiments.

\subsection{Numerical Experiments}\label{subsec:pushnum}

For all numerical experiments we perform, job sizes are assumed to be hyperexponentially distributed of order 2 (and mean 1). This distribution is uniquely defined through two parameters, the Squared Coefficient of Variation (SCV $\in [1,\infty)$) and a shape parameter $f \in [0,1]$ (see e.g.~also \cite{hellemansSIG18}).

In Figure \ref{fig:push} (left) we show the expected response times together with the lower and upper bounds obtained from Theorem \ref{th:qact}. Here we set $f=1/2$, $SCV=10$, $\delta \in \{0.15,0.5\}$ and $\lambda \in [0.5,1]$. We clearly see that decreasing $\delta$ or increasing $\lambda$ increases the mean response time and the values of the bounds. Note that the mean response time is non-differentiable at the values where $\tilde m\in \mathbb{N}$. Furthermore, the proposed bounds become exact at these points (as is clear from the proof of Theorem \ref{th:qact}).

In Theorem \ref{cor:push} we showed that $\tilde m$ depends on the job size distribution through the value of $y$ as defined in \eqref{eq:y}. In Figure \ref{fig:push} (right) we plot 
$\tilde m$ for $f=1/SCV$, $SCV \in \{5,10,50,250\}$, $\lambda \in [0.5,1)$ and $\delta = 0.5$. We observe that as the SCV increases, the gap with the upper bound reduces (to zero as $y$ converges to $1$).
This is however not true in general, for instance, when $f=1/2$ one finds that
$\tilde m$ does not approach the upper bound when the $SCV$ tends to infinity (as $y$ does not
converge to $1$ in such case).

\begin{figure}
    \centering
    \begin{minipage}{0.5\textwidth}
        \centering
        \includegraphics[width=1\textwidth]{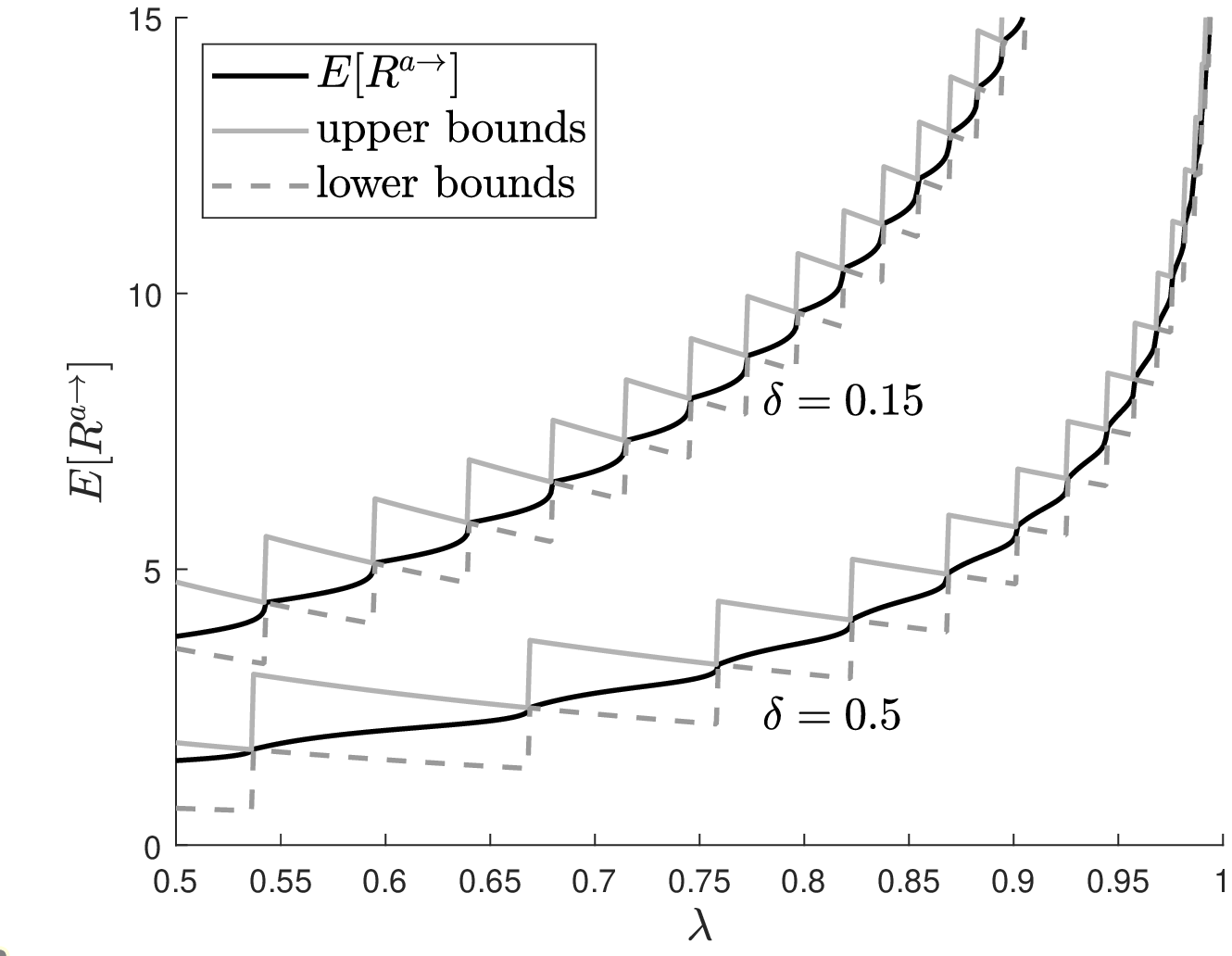}
    \end{minipage}\hfill
    \begin{minipage}{0.5\textwidth}
        \centering
        \includegraphics[width=1\textwidth]{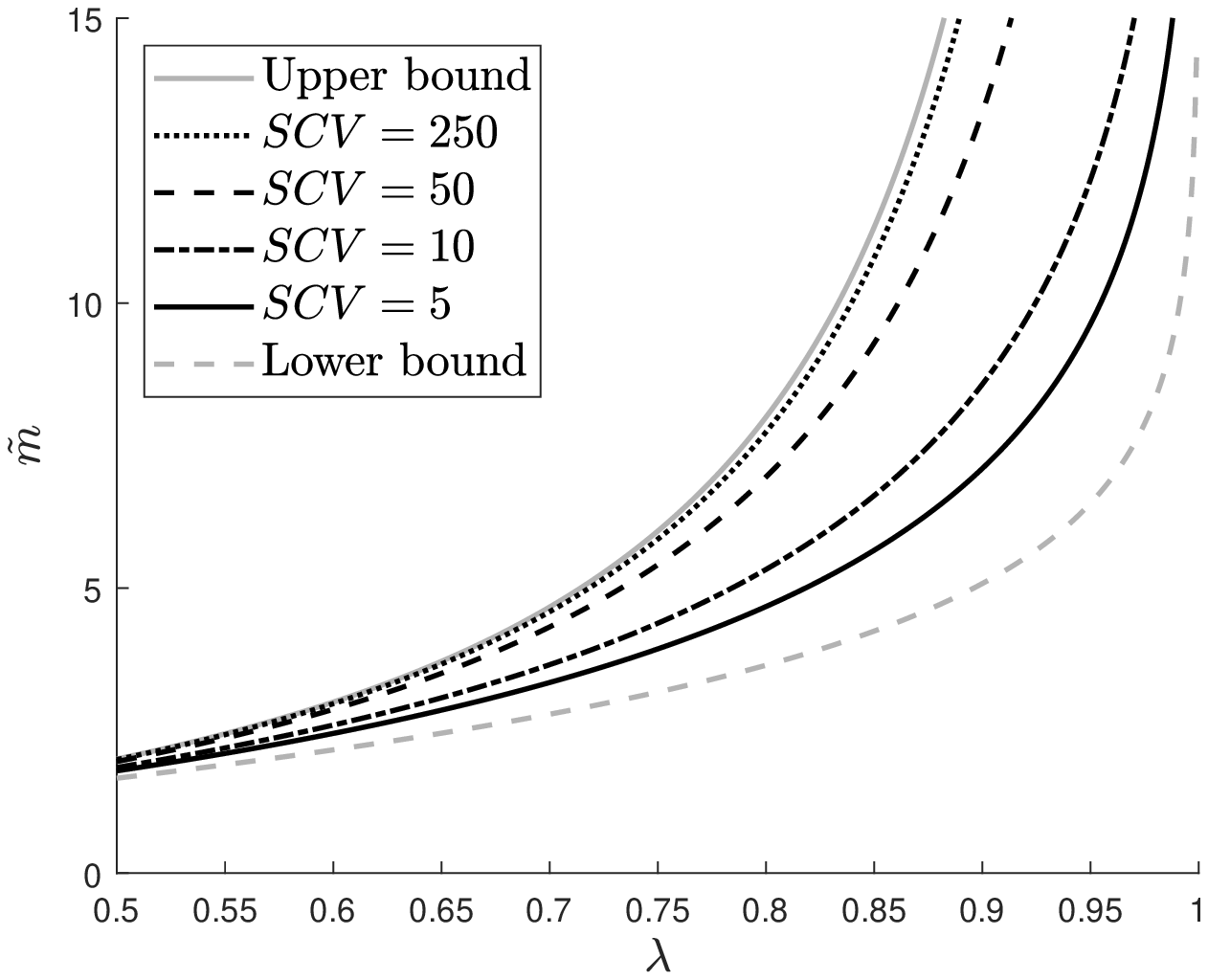}
    \end{minipage}
    \caption{Left: $E[R^{a{\shortrightarrow}}]$ in function of $\lambda$ with the lower and upper bounds from Theorem \ref{th:qact}. Right: $\tilde m$ in function of $\lambda$ for various values of the SCV with upper and lower bounds.}\label{fig:push}
\end{figure}

\section{Water filling} \label{sec:waterfilling}

In this section, we present the cavity approach for the \textit{water filling} policy introduced in \cite{ying2017power}. The accuracy of the cavity method for this policy is illustrated 
by simulation in Section \ref{app:simulation}. While the policy is quite different from the push policy, it turns out that its performance is quite similar. This similarity in performance was not even noted before in the exponential case.

Given a probe rate $\delta > 0$, each batch of jobs selects $\frac{\delta}{\lambda} M$ queues and the $M$ jobs are assigned using \textit{water filling} (with $M$ scaling as $\Theta(\log N)$). This entails that the overall probe rate is $\frac{\lambda}{M} \cdot \frac{\delta}{\lambda} M = \delta$. At any batch arrival, all selected queues are first filled up to some constant $m$ and some additional fraction of the selected servers get an additional arrival which raises their queue length to $m+1$. As the batch size scales with $N$, the cavity queue is characterized by two values, $m \in \mathbb{N}$ and $c \in [0,1]$. The cavity queue length jumps to $m$ at rate $\delta (1-c)$, while it jumps up to $m+1$ at rate $\delta c$. The state space of the queue at the cavity is therefore defined as:

\begin{equation}
	\Omega^w = \{0\} \cup \{(q,j) \mid q \in \{1,\dots,m+1\}, j \in \{1,\dots, n_s\}\},
\end{equation}
while the rate matrix is given by:
\begin{equation}\label{eq:water_filling_PH_Q}
	Q^w(m,c)
	=
	\begin{pmatrix}
		-\delta &  &  & & & \delta (1-c) \alpha & \delta c \alpha \\ 
		s^* & S-\delta I & & &  & \delta (1-c) I & \delta c I \\ 
		& s^* \alpha & S-\delta I & & & \delta (1-c) I & \delta c I \\ 
		&  & \ddots & \ddots & & \vdots & \vdots \\ 
		&  & & s^* \alpha & S-\delta I& \delta (1-c) I & \delta c I  \\ 
		&  &  &  & s^* \alpha & S - \delta c I & \delta c I \\ 
		&  &  &  & & s^* \alpha & S
	\end{pmatrix}.
\end{equation}
Let us denote by $\pi_k^w(m,c)$ the stationary probability that the queue length is equal to $k$ given the value of $m$ and $c$. In order to compute the stationary distribution we must
first determine $m$ and $c$ such that $\pi_0^w(m,c) = 1-\lambda$. We can again observe that by ordering the states lexicographically and applying \cite{busicMOR}[Theorem 1] that $\pi_0(m,c)$ is decreasing as a function of $c$.
Furthermore, setting $c = 0$, all states with a queue length of $m+1$ become transient, meaning that all queues have a queue length bounded by $m$. Setting $c = 1$, we observe that we always jump up to queue length $m+1$, this indicates that a system with parameters $(m+1, 0)$ is identical to a system with parameters $(m, 1)$. 

Combining these two observations, we find that if $m < m'$ or $m=m'$ and $c < c'$ we have: $\pi_0^w(m, c) > \pi_0(m', c')$. Therefore, there must exist a unique pair $(m, c)$ for each $\lambda < 1$ such that $\pi_0^w(m, c) = 1 - \lambda$.

For the \textit{push policy} we computed the value of $m$ by looking at the system with $\nu = 0$. Analogously, we can now look at the system with $c = 0$. Given the value of $\lambda$ and a PH distribution, we need to determine $m$ such that $\pi_0^w(m, 0) \geq \lambda 
\geq \pi_0^w(m,1)$. Taking a closer look, one observes that for $c = 0$ the transition matrix \eqref{eq:water_filling_PH_Q} is identical to the transition matrix for the push policy
with $\nu = 0$, see \eqref{eq:Qpush0}. This implies that the value of $m$ is given by $\lfloor \tilde m \rfloor$, with $\tilde m$ defined in \eqref{eq:push_m}. Therefore Propositions \ref{th:pim}, \ref{th:nu_eq_holds}, \ref{th:mbounds_orderk}, Theorems \ref{cor:push}, \ref{th:mbounds}, \ref{th:qact}  and Corollaries \ref{th:lam_m}, \ref{th:push_critical} also
hold for the water filling policy. Setting $y=\frac{1}{1 + \delta}$, we again find that $\tilde m = - \log(1-\lambda) / \log(1+\delta)$, which was also observed in \cite{ying2017power}[Theorem 3].

\begin{remark}
	The explicit formula for the stationary distribution in case of exponential job sizes in \cite{ying2017power}[Theorem 3] easily follows from setting $m = \lfloor \tilde m \rfloor$ and using the fact that $\pi Q^w(m,c) = 0$. Indeed, this yields the recursion $\pi_1^w(m,c) = \delta \pi_0^w(m,c)$ and $\pi_{k+1}^w(m,c) = (1+\delta) \pi_k^w(m,c)$ (for $ k \leq m-1$). Allowing us to conclude that $\pi_k^w(m,c) = (1+\delta)^{k-1} \delta (1-\lambda)$ for $k=1,\dots,m$. We can then compute:
	$$
	\pi_{m+1}^w(m,c) = 1-\sum_{k=0}^m \pi_k^w(m,c) = 1 - (1-\lambda) (1+\delta)^m.
	$$
\end{remark}

In order to compute the value of $c$ in case of PH job sizes, we have the following result:

\begin{theorem} \label{th:findc}
	For the \textit{water filling policy} with arrival rate $\lambda$, probe rate $\delta$ and $y$ as in \eqref{eq:y}, we have $m = \lfloor \tilde m \rfloor$ and $c \in [0,1)$ is the unique
	value such that:
	\begin{equation}
		1 - \lambda = \frac{y^{m-1}/\delta}{1/\delta + \kappa (-K)^{-1} \textbf{1}},
	\end{equation}
	with $\kappa = (1-c, c) \otimes (y^{m-1} \alpha + (1-y^{m-1}) \alpha')$, $\alpha' = \alpha (\delta I - S)^{-1} / \alpha (\delta I - S)^{-1} \textbf{1}$ and
	$$
	K = \begin{pmatrix}
		S - \delta c I & \delta I c\\
		\mu^* \alpha & S
	\end{pmatrix}.
	$$
\end{theorem}
\begin{proof}
	We first note that the mean time away from $\{m, m+1\}$ is simply given by $1/\delta$, as we jump back to $\{m, m+1\}$ at rate $\delta$ from any other state. Next, we compute the time we stay in $\{m, m+1\}$, this time can be described by a PH distribution with $2 \cdot n_s$ states, where the first $n_s$ states correspond to having queue length $m$ (and the other states are for queue length $m+1$).

	We jump up to state length $m$ with probability $1-c$ while we jump up to $m+1$ with probability $c$. This entails that the initial vector when we arrive in $\{m, m+1\}$ is indeed given by $\kappa$. It is clear that the transition matrix $K$ represents the transitions in $\{m,m+1\}$, we therefore find that the mean time spent in $\{m,m+1\}$ is given by $\kappa (-K)^{-1} \textbf{1}$.
	
	Combining these two observations we find that the mean cycle length is given by $1/\delta + \kappa (-K)^{-1} \textbf{1}$, and it remains to find the mean time we remain in $0$ in one cycle. To this end, we notice that a jump from an empty system occurs when we have had $m-1$ job completions since the last renewal, which happens with probability $y^{m-1}$. Moreover, the time we stay in zero is (on average) $1/\delta$. This yields the result.
\end{proof}

\begin{remark}
	When job sizes are exponential, we find that $\kappa = (1-c, c)$, 
	$$
	K = \begin{pmatrix}
		-1 - c\delta & c \delta\\
		1 & -1
	\end{pmatrix},
	$$
	and $y = \frac{1}{1 + \delta}$. From this, it is not hard to see that we recover 
	the formula in \cite{ying2017power}:
	$$
	c = \frac{1}{\delta (1-\lambda)(1+\delta)^m} - \frac{1}{\delta}.
	$$
\end{remark}

\begin{theorem} \label{th:main_th_waterfilling}
	In the same setting as Theorem \ref{th:findc}, we find that
	$\pi_0^w(m,c)=1-\lambda$ and
	\begin{align}\label{eq:piw} 
		\pi^w_q(m,c) \textbf{1} = (1-\lambda) (1/y - 1)/ y^{q-1},
	\end{align} 
	for $q = 1, \ldots,m-1$. Further, 
	\begin{align}\label{eq:piw_mp1}  
		\pi_{m+1}^w(m,c) \textbf{1} =  1 - \frac{1}{c}\left( \lambda/\delta - \sum_{q=0}^{m-1} (m-q) \pi_q^w(m,c)\textbf{1}\right),
	\end{align} 
	and
	\begin{align}\label{eq:piw_m} 
		\pi_m^w(m,c) \textbf{1} =  1 - (1-\lambda)y^{1-m} - \pi_{m+1}^w(m,c)  \textbf{1}.
	\end{align} 
\end{theorem}
\begin{proof}
	Consider the chain with rate matrix $Q^w(m,c)$ censored on the states with $q < m$.
	Let $S_{<i}$ be the set of states with $q < i$. Define renewal cycles for
	this censored chain such that the start of a cycle corresponds to the points in time
	that the original chain makes a jump from a state with $q=m$ to a state with $q=m-1$.

	The probability that the set $S_{<i}$ is reached during a cycle is clearly given
	by $y^{m-i}$ and the mean time that the censored chain stays in the set $S_{<i}$
	given that the set is reached 
	equals $1/\delta$. Note that the mean cycle length for the censored chain also
	equals $1/\delta$. This implies that
	\begin{align}\label{eq:piw1}  
		\pi_0^w(m,c) + \sum_{j=1}^{i-1} \pi_j^w(m,c) \textbf{1} = y^{m-i} (1-\pi_m^w(m,c)\textbf{1}-\pi_{m+1}^w(m,c)\textbf{1}). 
	\end{align}
	As $m$ and $c$ are such that $\pi_0^w(m,c)=1-\lambda$, the above with $i=1$ yields
	\begin{align}\label{eq:piw2} 
		1-\lambda = y^{m-1} (1-\pi_m^w(m,c)\textbf{1}-\pi_{m+1}^w(m,c)\textbf{1}),
	\end{align}
	which implies \eqref{eq:piw_m}. Combining \eqref{eq:piw1} and \eqref{eq:piw2}
	shows that
	\[ \pi_0^w(m,c) + \sum_{j=1}^i \pi_j^w(m,c) \textbf{1} = (1-\lambda)/ y^{i}, \]
	and \eqref{eq:piw} follows.
	Finally, \eqref{eq:piw_mp1} follows from 
	\[ \lambda = \delta \left( \sum_{i=0}^{m} (m-i) \pi_i^w(m,c)\textbf{1} + (1-\pi_{m+1}^w(m,c)\textbf{1}) c \right).\]
	In this equality, the left hand side corresponds to the total number of arrivals per unit of time, while the right hand side signifies the number of jobs assigned to servers per unit of time. Therefore, the equality can be proven in the same way as Proposition \ref{th:nu_eq_holds}.
\end{proof}

\subsection{Numerical Experiments}

In Figure \ref{fig:water} (left), we depict $E[R^w]$ as a function of $\delta$. We set $f=1/2$, $SCV = 10$, $\delta \in [0.3,1.2]$ and $\lambda \in \{0.6,0.7,0.8,0.9\}$. Clearly, increasing $\lambda$ or decreasing $\delta$ increases $E[R^w]$. We observe the same type of irregular behaviour as in Figure \ref{fig:push} (left), that is, the curve becomes non-differentiable at the values of $\delta$ for which $\tilde m \in \mathbb{N}$.

In Figure \ref{fig:water} (right) we illustrate the influence of $y$ on the mean response time. To this end we use the hyperexponential distributions $Z(\varepsilon)$ which was introduced in the proof of Theorem \ref{th:mbounds} (which also holds for the water filling strategy). We set $\delta = 0.5$ and $\lambda \in \{0.6,0.7,0.8,0.9\}$. With $\delta = 0.5$ we find that $y$ ranges from $2/3$ (for $\varepsilon = 1/2$) to 1 (for $\varepsilon \rightarrow 0^+$). As $y$ gets close to 1 the frequency of sudden increases in $E[R^w]$ increases. This is due to the fact that the maximal queue length increases more often as $y$ gets close to 1. However, we observe that the limiting value for $\E[R^w]$ with $y=1$ is still finite.

\begin{figure}
    \centering
    \begin{minipage}{0.5\textwidth}
        \centering
        \includegraphics[width=1\textwidth]{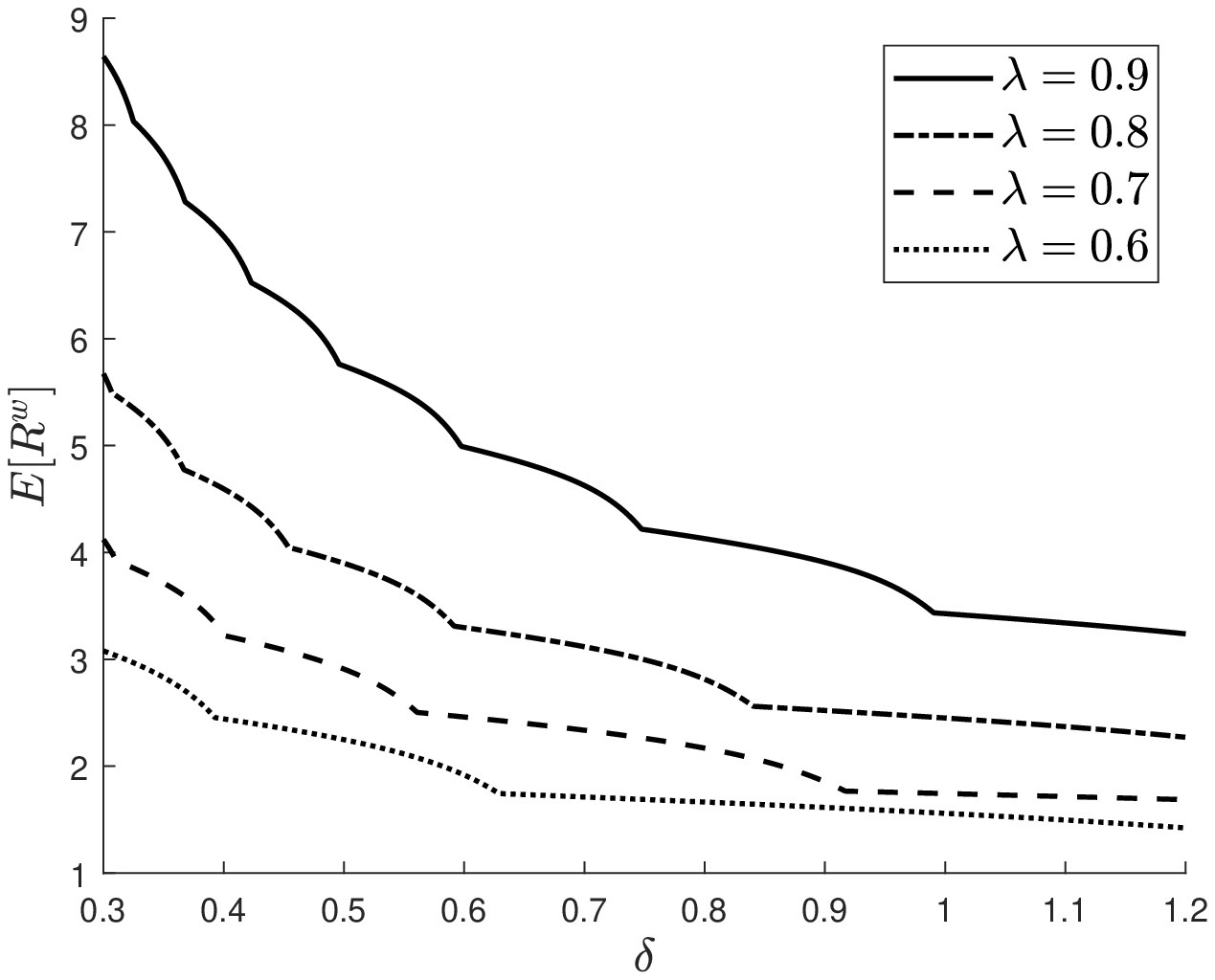}
    \end{minipage}\hfill
    \begin{minipage}{0.5\textwidth}
        \centering
        \includegraphics[width=1\textwidth]{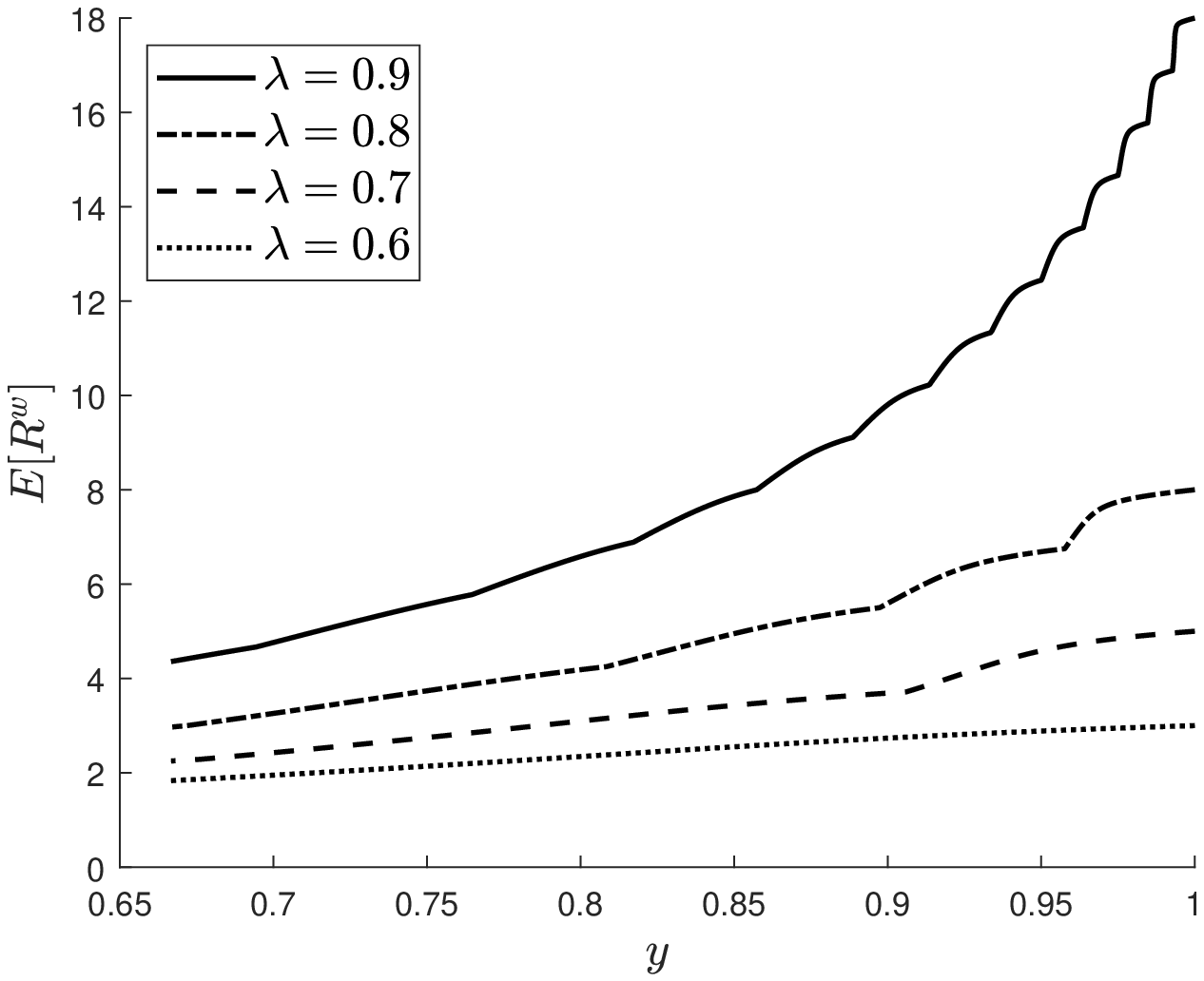}
    \end{minipage}
    \caption{$E[R^w]$ in function of $\lambda$ and $\delta$ (left) and in function of $\lambda$ and $y$ (right).}\label{fig:water}
\end{figure}

\section{Hyperscalable pull policy}\label{sec:pull}
In this section we study the queue at the cavity for the pull policy. Simulation results
that study the accuracy of the cavity method are presented in Section \ref{app:simulation}.  Recall that a server updates the dispatcher with its
current queue length information with probability $\delta_1$ when it completes service
of a job and at rate $\delta_0$ when it is idle. As the mean service time of a job is
equal to one and $1-\lambda$ is the faction of time that a server is idle, this means that the overall update rate is given by $\delta = \lambda \delta_1 + (1-\lambda)
\delta_0$. 

Given $\delta < \lambda$, the range of $\delta_1$ is given by $(0,\delta/\lambda)$.
When $\delta \geq \lambda$, we can set $\delta_1=1$ such that servers always update at service
completion times. This  implies that this policy reduces to the Join-Idle-Queue policy,
which has vanishing wait. If $\delta=\delta_0=\delta_1$, the overall update rate automatically
equals $\delta$, which means that there is no need for servers to know the arrival rate $\lambda$. However, when $\delta_1 \not= \delta$, then $\lambda$ must be known in order to set $\delta_0$
such that the overall update rate equals $\delta$.  

As jobs are assigned in a greedy manner based on the estimated queue lengths, we again find that
in the large-scale limit, all servers have an estimated queue length equal to $m$ or $m+1$ for some integer $m \geq 0$ and the state space for the queue at the cavity is the same as for the push policy, that is,  
\[ \Omega^{\shortleftarrow} = \{(0,m),(0,m+1)\} \cup \{(q,e,j) | e = m,m+1 ; q = 1,\ldots,e; j = 1,\ldots,n_s \},\]
where $e$ is the estimated queue length, $q$ the actual queue length and $j$ the service phase.
The rate matrix $Q^{\shortleftarrow}(m,\nu)$ for the pull policy has a similar structure as the rate
matrix $Q^{\shortrightarrow}(m,\nu)$ given by \eqref{Qpush}, where we replace the right arrows by
left arrows to indicate that we are discussing the pull policy.
For the pull policy, a service completion only leads to a decrease in the actual queue
length if the service completion is not accompanied by an update, thus
$Q^{\shortleftarrow}_{q,q-1} = (1-\delta_1) Q^{\shortrightarrow}_{q,q-1}$, for $q=1,\ldots,m$,
and $Q^{\shortleftarrow}_{m+1,m} = [s^*\alpha\delta_1 \ \ \ s^*\alpha (1-\delta_1)]$. 

If an update does occur at a service completion time, $(q,e)$ becomes $(m,m)$ similar
to a probe event for the push policy, hence
\[Q^{\shortleftarrow}_{q,m} = \begin{bmatrix}
	\delta_1 s^*\alpha & 0 \\ \delta_1 s^*\alpha & 0
\end{bmatrix},\mbox{ \ \ and \ \   }   
Q^{\shortleftarrow}_{0,m} = \begin{bmatrix}
	\delta_0 \alpha & 0 \\ \delta_0 \alpha & 0
\end{bmatrix},\] 
for $0 < q < m-1$. Arrivals that are
assigned to a server with an estimated queue length equal to $m$ still occur at some
rate $\nu$, hence $Q^{\shortleftarrow}_{q-1,q}(\nu)=Q^{\shortrightarrow}_{q-1,q}(\nu)$, for $q\not= m-1$
and
\[
Q^{\shortleftarrow}_{m-1,m}(\nu) = \begin{bmatrix}
	\delta_1 s^*\alpha & \nu I \\ \delta_1 s^*\alpha & 0
\end{bmatrix}.\] 
Note that \eqref{eq:nu_check_push} is no longer valid for the rate $\nu$. Instead we have
\begin{align}\label{eq:nu_check_pull} \nu \pi^{e\shortleftarrow}_m(m,\nu) = \lambda - 
	\delta_0 m \pi^{a\shortleftarrow}_0(m,\nu) - \delta_1 
	\sum_{q=1}^m \sum_{e=m}^{m+1} \sum_{j=1}^{n_s}  (m-q+1) \pi_{(q,e,j)}^{\shortleftarrow}(m,\nu) s^*_j ,
\end{align}
where  $\pi_{(q,e,j)}^{\shortleftarrow}(m,\nu)$ is the steady state probability to be in state $(q,e,j)$,
as idle servers update at rate $\delta_0$ and an update adds $m$ jobs to the server,
while a busy server with $q$ jobs in phase $j$ completes service and updates
at rate $s^*_j \delta_1$ and adds $m-q+1$ jobs to the server.
The proof of \eqref{eq:nu_check_pull} is similar to that of Proposition \ref{th:nu_eq_holds}.

The diagonal blocks capture changes in the service phase, thus
\[Q^{\shortleftarrow}_{q,q}(\nu) = \begin{bmatrix}
	S-\nu I & 0 \\ 0 & S 
\end{bmatrix}, \ \
Q^{\shortleftarrow}_{m,m}(\nu) = \begin{bmatrix}
	S-\nu I & 0 \\ \delta_1 s^* \alpha & S
\end{bmatrix}, \mbox{  \ \ and \ \    }   
Q^{\shortleftarrow}_{0,0}(\nu) = \begin{bmatrix}
	-(\nu+\delta_0) & 0 \\ 0 & -\delta_0 
\end{bmatrix},
\] 
for $q=1,\ldots,m-1$
and $Q^{\shortleftarrow}_{m+1,m+1} = Q^{\shortrightarrow}_{m+1,m+1} = S$.


\subsection{Finding $m$ and $\nu$}
To assess the performance of the queue at the cavity we need to determine the
unknowns $m$ and $\nu$. As in the push case, we can 
find $m$ by
studying the Markov chain with $\nu=0$ characterized by $Q^{\shortleftarrow}(m,0)$
and using a bisection algorithm to set $\nu$ once $m$ is known. 
When $\nu = 0$ the states with $e=m+1$ are transient and 
we can remove these states such that this chain evolves on the state space
$\Omega_{(0)}^{\shortleftarrow} = \Omega_{(0)}^{\shortrightarrow}$ and has rate matrix $Q^{\shortleftarrow}_{(0)}(m)$ given by
\[
Q^{\shortleftarrow}_{(0)}(m) = \begin{bmatrix}
-\delta_0  & & & & & \delta_0 \alpha \\
(1-\delta_1) s^* & S & &  & & \delta_1 s^* \alpha\\
 & (1-\delta_1) s^*\alpha & S &  & & \delta_1 s^* \alpha \\
&  &\ddots& \ddots & & \vdots   \\
& & & (1-\delta_1) s^*\alpha & S & \delta_1 s^* \alpha \\
&  &&  & (1-\delta_1) s^* \alpha &  S + \delta_1 s^* \alpha
\end{bmatrix}.
\]

\begin{proposition}\label{th:pim2}
The steady state probabilities of $Q^{\shortleftarrow}_{(0)}(m)$ are such that for $i=1,\ldots,m+1$
\begin{align}\label{pim2}
\sum_{q=0}^{i-1} \pi^{a\shortleftarrow}_q(m) = \frac{\delta_0(1-\delta_1)^{m-i+1} +(1-\delta_1)^m (\delta_1-\delta_0)}{\delta_0 + (1-\delta_1)^m (\delta_1-\delta_0)}. \end{align}
\end{proposition}
\begin{proof}
We define a renewal cycle in the same manner as in the proof of Proposition \ref{th:pim}, that is,
a cycle starts whenever the chain leaves level $m$. The mean time in level $m$ is now
the same as the mean service time and thus equal to one. The mean time in states
of the form $(q,j)$ for  $0 < q < m$ per cycle is given by $(1-\delta_1)^{m-q}$, while
the mean time in state $0$ per cycle equals $(1-\delta_1)^{m}/\delta_0$.
This implies that the mean cycle length equals
\[\frac{1}{\delta_1} + (1-\delta_1)^m \left(\frac{1}{\delta_0}-\frac{1}{\delta_1}\right),\]
and the mean time in states with $q < i$ per cycle equals
\[ (1-\delta_1)^m/\delta_0 + \sum_{q=1}^{i-1} (1-\delta_1)^{m-q},\]
which yields the result.
\end{proof}

When $\delta_1=\delta_0=\delta$, the right hand side of \eqref{pim2} simplifies to $(1-\delta)^{m-i+1}$, while
letting $\delta_1$ tend to zero reduces it to $(i-1+1/\delta_0)/(m+1/\delta_0)$.
Setting $i=1$ in the previous result implies the following:
\begin{theorem}\label{cor:pull}
	For the \textit{pull policy} with arrival rate $\lambda \in [0,1)$, probe probability $\delta_1$ at job completions and probe rate $\delta_0$ at idle servers, we have $1 -  \lambda \in [\pi^{a\shortleftarrow}_0(\lceil \tilde m \rceil),
\pi^{a\shortleftarrow}_0(\lfloor \tilde m \rfloor)]$ for
\begin{align}\label{eq:pull_m}
\tilde m = \left.\log\left(1-\lambda \delta_1/\delta \right)\middle/\log(1-\delta_1)\right.,
\end{align}
with $\delta = \lambda \delta_1 + (1-\lambda)\delta_0$ the overall update rate.
Hence, $\lceil \tilde m \rceil$ represents the maximum queue length for the queue at the
cavity.
\end{theorem}
\begin{proof}

When $\pi^{a\shortleftarrow}_0( m)=1-z\in (0,1)$, we have due to Proposition \ref{th:pim2}
\begin{align}\label{eq:pull_lam}
z = 1 - \frac{\delta_1}{\delta_0/(1-\delta_1)^{m} + (\delta_1 - \delta_0)},
\end{align}
which shows that $z$ increases as a function of $m$ and equals $\lambda$ for $m=\tilde m$.
\end{proof}

There are a number of interesting observations we can make based on this result:
\begin{enumerate}
\item The maximum queue length $\lceil \tilde m \rceil$ is \textit{insensitive} to
the job size distribution and whenever $\lambda$ is such that it is equal to
the right hand side of \eqref{eq:pull_lam}
for some integer $m$, the entire queue length distribution is insensitive to the
job size distribution.
\item  The derivative of $\tilde m$ with respect to $\delta_1$ is given
by 
\[\frac{d \tilde m}{d \delta_1} = \frac{\delta_1}{\delta \log(1-\delta) (\lambda
\delta_1/\delta - 1)} < 0, \]
for $\delta_1 \in (0,\delta/\lambda)$.
Therefore, the maximum queue length is minimized by 
setting $\delta_1 = 0$, that is, letting only idle servers
update at rate $\delta_0 = \delta/(1-\lambda)$. As
\[ \lim_{\delta_1 \shortrightarrow 0^+} \tilde m =  \lim_{\delta_1 \shortrightarrow 0^+}
 \left.\log\left(1- \lambda \delta_1/\delta \right)\middle/\log(1-\delta_1)\right. = 
\lambda/\delta,\]
we find that the maximum queue length simply reduces to $\lceil \lambda/\delta \rceil$.
\item When only the idle servers send updates, the rate $\delta_0$ must be set
equal to $\delta/(1-\lambda)$, which indicates that the arrival rate $\lambda$
must be known in order to achieve a target overall update rate $\delta$. Setting
$\delta = \delta_0 = \delta_1$ does not require knowledge of the arrival rate and
results in a maximum queue length of $\lceil \tilde m \rceil$ with
\[\tilde m = \log\left(1- \lambda \right)/\log(1-\delta). \]
\end{enumerate}

\begin{corollary}\label{th:lam_m_pull}
In the same setting as Theorem \ref{cor:pull}, the maximum queue length of the queue at the cavity is equal to $m > 0$ for $\lambda \in (\lambda^{\shortleftarrow}_{m-1},\lambda^{\shortleftarrow}_m]$
with 
\begin{align}\label{eq:lam_m_pull}
\lambda^{\shortleftarrow}_m = \frac{\delta_0 -\delta_0  (1-\delta_1)^m}
{\delta_0 - \delta_0 (1-\delta_1)^m + \delta_1 (1-\delta_1)^m}.
\end{align}
\end{corollary} 
\begin{proof}
The result is immediate by \eqref{eq:pull_lam} as the maximum queue length increases by
one whenever $\lambda$ is such that $\pi^{a\shortleftarrow}_0(m) = 1-\lambda$ 
for some integer $m$.
\end{proof}
When $\delta_0=\delta_1$, we have $\lambda^{\shortleftarrow}_m = 1-(1-\delta)^m$.
For $\delta_1$ tending to zero we find
\[ \lim_{\delta_1 \shortrightarrow 0^+} \lambda^{\shortleftarrow}_m  = 1-\frac{1}{\delta_0 m +1},\]
with $\delta_0 = \delta/(1-\lambda^{\shortleftarrow}_m)$. This means that if only
idle servers pull we have $\lambda^{\shortleftarrow}_m = \delta m$, which is in agreement
with the fact that the maximum queue length is bounded by $\lceil \lambda/\delta \rceil$.

\subsection{Performance Bounds}
As the maximum queue length $\lceil \tilde m  \rceil$ is insensitive 
to the job size distribution for the pull policy, there is no result similar
to Theorem \ref{th:mbounds}. However, we do obtain bounds on the average queue length.

\begin{theorem} \label{th:bounds_EQ_pull}
In the same setting as Theorem \ref{cor:pull}, let $E[Q^{a\shortleftarrow}]$ be the mean queue length of the queue at the cavity, then
$q^{\shortleftarrow}(\lfloor \tilde m \rfloor) \leq E[Q^{a\shortleftarrow}] \leq q^{\shortleftarrow}(\lceil \tilde m \rceil)$, where
$\tilde m$ is given by \eqref{eq:pull_m} and
\begin{align}
q^{\shortleftarrow}(m)=\frac{\delta_0(m+1)-\delta_0(1-(1-\delta_1)^{m+1})/\delta_1}
{\delta_0+(1-\delta_1)^m(\delta_1-\delta_0)}.
\end{align}
\end{theorem} 
\begin{proof}
The proof is identical to the proof of Theorem \ref{th:qact}, except that we use
\eqref{pim2} instead of \eqref{pim}. 
\end{proof}
\begin{remark}
In the special case that $\delta_0=\delta_1=\delta$ we find that
\[q^{\shortleftarrow}(m) = (m+1)-(1-(1-\delta)^{m+1})/\delta = 
(m+1)-\lambda_{m+1}^{\shortleftarrow}/\delta.\]
as $\lambda_m^{\shortleftarrow} = 1-(1-\delta)^m$ in that case.

On the other hand, when $\delta_1$ tends to zero, we have
\[q^{\shortleftarrow}(m) = \frac{m(m+1)}{2} \frac{\delta_0}{\delta_0 m +1}
= 
\frac{(m+1) \lambda_m^{\shortleftarrow} }{2(1+ \lambda_m^{\shortleftarrow} -\lambda)},\]
as $\delta_0=\delta/(1-\lambda)$ and $\lambda_m^{\shortleftarrow} = \delta m$.
\end{remark}

\subsection{Critically loaded system}\label{subsec:pullcrit}
The limit $\lambda \shortrightarrow 1^-$ heavily depends on the chosen value for $\delta_1$. Indeed the scaling we require is given by $\log\left( \frac{1}{1-\lambda \frac{\delta_1}{\delta}} \right)$. In particular, if $\delta_1 = 0$, we find that the maximum queue length simply converges to $1/\delta$ for $\lambda \shortrightarrow 1^-$, meaning no scaling is required at all. For $\delta_1 > 0$, the limit we obtain with the proper scaling is given by: 
$$
\lim_{\lambda \shortrightarrow 1^-} \frac{E[R^{a\shortleftarrow}]}{\log\left( \frac{1}{1-\lambda\delta_1/\delta} \right)}=
\lim_{\lambda \shortrightarrow 1^-} \frac{E[Q^{a\shortleftarrow}]}{\log\left( \frac{1}{1-\lambda\delta_1/\delta} \right)}=
\lim_{\lambda \shortrightarrow 1^-} \frac{\tilde m}{\log\left( \frac{1}{1-\lambda\delta_1/\delta} \right)} = \frac{1}{\log\left(\frac{1}{1-\delta_1}\right)}.
$$ 
The proof of this statement is similar to the proof of Theorem \ref{th:push_critical}, except 
that the last equality follows directly from Corollary \ref{cor:pull}. 

\subsection{Numerical Experiments}\label{subsec:pullnum}
In Figure \ref{fig:pull} (left) we set $f=1/2$, $SCV=10$, $\delta \in \{0.15,0.5\}$, $\delta_1=0$ (i.e.~only idle servers pull) and $\lambda \in [0.5,1]$. The expected response times together with lower and upper bounds obtained from Theorem \ref{th:bounds_EQ_pull} are shown. Further, as $\lambda \rightarrow 1^-$, the mean response time stays finite, as was noted in Subsection \ref{subsec:pullcrit}.


In Figure \ref{fig:pull} (right) we plot $\frac{E[R^{{a}\shortleftarrow}]-1}{\tilde m}$ in function of $\lambda$  for $f=1/2$, $SCV=10$, $\delta \in \{0.15,0.5,0.7\}$, $\delta_1 = \delta$ and $\lambda \in \{0.1, 1-10^{-5}\}$. Note that $E[R^{{a}\shortleftarrow}]-1$ is the mean waiting
time and due to the bounds on $E[R^{{a}\shortleftarrow}]$, we know that the ratio $(E[R^{{a}\shortleftarrow}]-1)/\tilde m$ converges to one. 

\begin{figure}
    \centering
    \begin{minipage}{0.5\textwidth}
        \centering
        \includegraphics[width=1\textwidth]{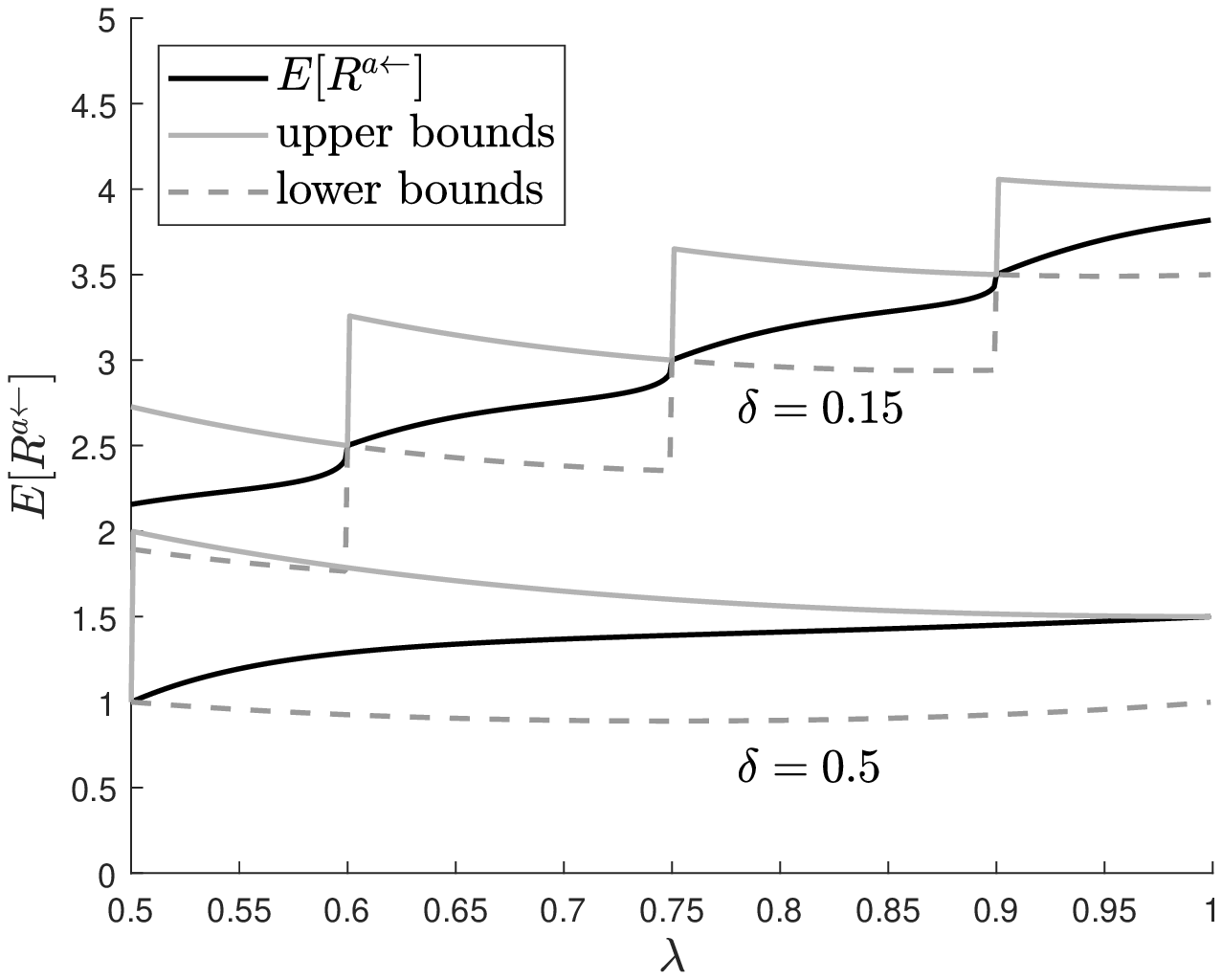}
    \end{minipage}\hfill
    \begin{minipage}{0.5\textwidth}
        \centering
        \includegraphics[width=1\textwidth]{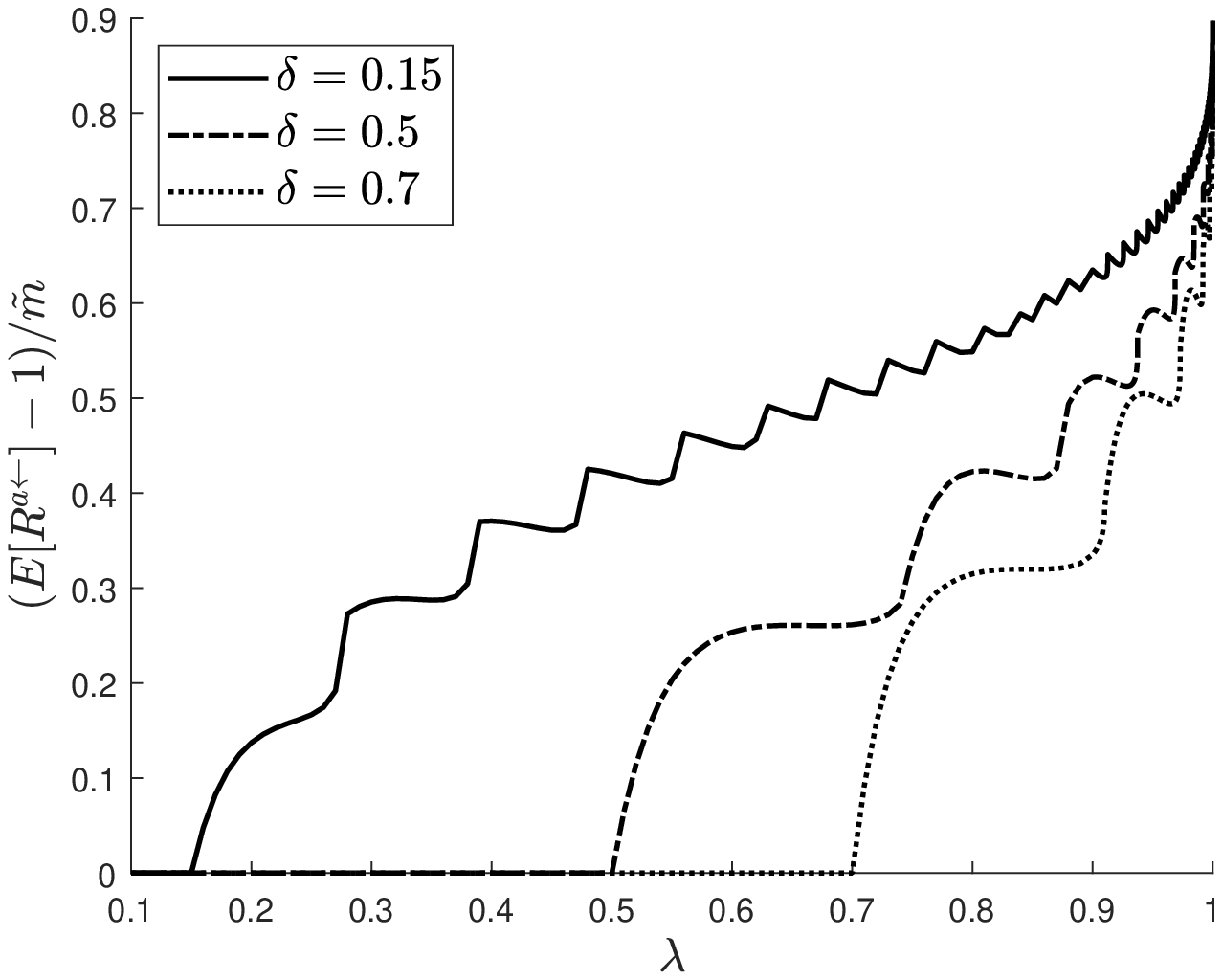}
    \end{minipage}
	\caption{$E[R^{a{\shortleftarrow}}]$ with lower and upper bounds (left) and $(E[R^{{a}\shortleftarrow}]-1)/\tilde m$ (right) in function of $\lambda$ and $\delta$.}\label{fig:pull}
\end{figure}

\section{On the power of (even a little) resource pooling} \label{sec:pooling}
In this section we study the queue at the cavity for the resource pooling policy of \cite{tsitsiklis2013power}. As for the other policies, simulation results
that demonstrate accuracy of the cavity method are presented in Section \ref{app:simulation}. When $\lambda > p$, the rate at which jobs leave the system is given by $(1-\pi_0) (1-p) N + pN$, with $\pi_0$ the fraction of idle servers. As the number of incoming jobs must equal the number of outgoing jobs this entails:
\begin{equation}
	(1-\pi_0) (1-p) + p = \lambda \Rightarrow \pi_0 = \frac{1-\lambda}{1-p}.
\end{equation}
In case $p \geq \lambda$, all the jobs are processed by the central server and
the cavity queue is idle with probability one. We generalize the analysis for exponential job sizes in \cite{tsitsiklis2013power} to the case of PH job sizes for $p < \lambda$. To this end, we note that the cavity queue is similar to an $M/PH/1$ queue with a maximal queue length given by $m+1$ and an adjusted departure rate from level $m+1$ to $m$, where $m$
depends on $\lambda$, $p$ and the job size distribution. Indeed, in the
large-scale limit the fraction of servers with more than $m+1$
jobs equals zero due to the presence of the centralized server.
As part of the capacity of the centralized server is consumed by processing jobs
that arrive in a queue with a length $> m$, the remaining capacity results in 
an additional service rate $\omega$ when the queue length equals $m+1$.
 In other words, as with the previous policies we have two unknowns: $m$ and $\omega$.
The state space for the cavity queue is given by:
\begin{align*}
	\Omega^r &= \{ 0 \} \cup \{ (q, j) \mid q=1,\dots,m+1 ; j=1,\dots, n_s \}.
\end{align*}

\subsection{State transitions}
For all queue lengths $q \leq m$, the job in service simply receives service at rate $(1-p)$. However, when the queue length $q$ equals $m+1$, a job from the cavity queue is
selected by the central server at some rate $\omega$ as noted above.
Denote by $\pi_{q}^r(m, \omega)$ the probability that the cavity queue has length $q$ given $m$ and $\omega$.
It is not hard to see that the rate $\omega$ must obey:
\begin{equation}
	\omega = p N \cdot  \frac{1}{N \pi_{m+1}^r(m,\omega)} \left( 1 - \frac{\lambda \pi_{m+1}^r(m,\omega) \cdot N}{p N} \right) = \frac{p - \lambda \pi_{m+1}^r(m,\omega)}{\pi_{m+1}^r(m,\omega) },
\end{equation}
as the centralized server generates tokens at a rate equal to $pN$, has a probability of $1/(N \pi_{m+1}^r(m,\omega))$ to pick the cavity queue given that it has length $m+1$ and the fraction of tokens devoted to queues with length $m+2$ is given by $\frac{\lambda \pi_{m+1}^r(m,\omega) \cdot N}{p N}$. We therefore find that the  matrix for the cavity queue is defined as \eqref{Qpooling}.

\begin{align}\label{Qpooling}
	&Q^{r}(m,\omega) = \nonumber \\
	& \ \ \ \begin{bmatrix}
		-\lambda & \lambda \alpha &  & & & \\
		(1-p) s^* & (1-p)S - \lambda I & \lambda I &  & & \\
		& (1-p) s^* \alpha & (1-p) S - \lambda I & \lambda I && & \\
		&  &\ddots & \ddots & \ddots & \\
		&  & & (1-p) s^* \alpha & (1-p) S - \lambda I & \lambda I \\
		&  & & & (1-p) s^* \alpha+\omega I & (1-p)S-\omega I
	\end{bmatrix}.
\end{align}

Note that when $\omega = 0$ (or $\omega=\infty$) this rate matrix is identical to the rate matrix
of a bounded M/PH/1 queue with room for $m+1$ (or $m$) jobs.

\subsection{Finding $m$ and $\omega$}
In order to analyze this policy, we should determine the value of the $m$ and $\omega$ parameters. We can again make use of \cite{busicMOR}[Theorem 1] to argue that the probability 
to have an idle   cavity queue (that is $\pi_0^r(m,\omega)$) is decreasing as a function of $m$ and increasing as a function of $\omega$. Here we find that as $\omega$ increases from $0$ to infinity, the value of $m$ jumps down by one, $\omega=0$ corresponds to having no additional transitions from $m+1$ to $m$, while $\omega=\infty$ means that the additional
transition rate from $m+1$ to $m$ is infinite making the states $(m+1,j)$ transient.

From these observations, we find that $\pi_0^r(m,\omega) > \pi_0^r(m', \omega')$ if $m < m'$ or $m=m'$ and $\omega > \omega'$, which implies the existence of a unique ($m, \omega$) for which $\pi_0^r(m, \omega) = (1-\lambda)/(1-p)$. We first derive a method which can be used to compute $m$, using the correct $m$ value we indicate how to compute $\omega$ and therefore also the stationary distribution of the cavity queue. We also show that our method allows us to recover the results for exponential job sizes presented in \cite{tsitsiklis2013power}.

In order to compute $m$, we may assume that $\omega=\infty$ and therefore $(m+1,j)$ are transient states. We can further restrict our attention to the case with $p < \lambda$, otherwise $m=0$
as noted earlier.
The rate matrix $Q^{r}(m,\infty)$ is identical to that of an M/PH/1/m queue and we can
therefore use the results in \cite[Section 3.2]{neuts2} to express its steady state
probabilities  as follows:
\begin{align}
	\pi_q^{r}(m,\infty) &= \pi_0^{r}(m,\infty) \alpha R^q,\\
	\pi_m^{r}(m,\infty) &= \pi_0^{r}(m,\infty) \alpha R^{m-1} (-\lambda ((1-p)S)^{-1}),\\
	\pi_0^{r}(m,\infty) &= \left( \alpha \left[ \sum_{i=0}^{m-1} R^i + \lambda R^{m-1} (-(1-p)S)^{-1} \right]\textbf{1} \right)^{-1},
\end{align}
for $q=1,\ldots,m-1$, where
\[R = \lambda (\lambda I - (1-p)S -\lambda \textbf{1} \alpha)^{-1}.\]
As $\pi_0^{r}(m,\infty)$ decreases as a function of $m$, 
and  $\pi_0^{r}(0,\infty) = 1 > (1-\lambda)/(1-p) $ (for $p < \lambda$), the value of $m$ is found
as the largest $m$ such that $\pi_0^{r}(m,\infty) > (1-\lambda)/(1-p)$.
In other words it is the smallest $m$ such that $\pi_0^{r}(m+1,\infty) < (1-\lambda)/(1-p)$.

For exponential job sizes the matrix $R$ becomes a scalar equal to $\lambda/(1-p)$ and
$\pi_0^{r}(m,\infty)$ simplifies to $(1-\lambda/(1-p))/(1-(\lambda/(1-p))^{m+1})$.
Solving $\pi_0^{r}(m,\infty) = (1-\lambda)/(1-p)$ yields that
$m+1 = \log(p/(1-\lambda))/\log(\lambda/(1-p))$, which is in agreement with the result presented in \cite{tsitsiklis2013power}. Unfortunately, for PH job sizes no simple explicit formula for $m$ seems to exist,
in contrast to the push, water filling and pull policies studied in this paper.

Having computed $m$, the unique value of $\omega$ can now be determined using a bisection algorithm as the steady state probability $\pi_0^r(m,\omega)$ increases as $\omega$ increases
and should match $(1-\lambda)/(1-p)$. Due to the structure of the rate matrix
$Q^r(m,\omega)$, its stationary distribution can be computed in $O(mn_s^3)$ time using the algorithm in 
\cite{gaver1}. We can however do even better using the following result:

\begin{theorem} \label{th:resource_pooling1}
	For the resource pooling policy with $0 < p <\lambda < 1$, we find that when $\omega$ is set such that $\pi_0^r(m,\omega)=(1-\lambda)/(1-p)$, one finds
	that $\pi^r_q(m,\omega)$ is independent of $\omega$ and given by
	\begin{align}\label{eq:pir} 
		\pi^r_q(m,\omega) = \frac{1-\lambda}{1-p} \alpha R^q,
	\end{align} 
	for $q = 1, \ldots,m-1$. Further, compute 
	\begin{align}\label{eq:omega}  
		(\pi^r_m(m,\omega),\pi^r_{m+1}(m,\omega)) = (-\lambda \pi_{m-1}^r(m,\omega),0)
		\begin{bmatrix}
			(1-p)S-\lambda I & \lambda I \\ (1-p)s^* \alpha + \omega I & (1-p)S-\omega I
		\end{bmatrix}^{-1},
	\end{align} 
	then $\omega$ is the unique value such that $\pi_0^r(m,\omega)+\sum_{q=1}^{m+1} \pi^r_q(m,\omega)\textbf{1}
	=1$.
\end{theorem}
\begin{proof}
	The result is immediate from the structure of $Q^r(m,\omega)$ and the fact that
	the chain when censored on the states with $q < m$ is identical to an ordinary
	M/PH/1 queue censored on the states with $q < m$.
\end{proof}

When the job sizes are exponential, we have $R=\lambda/(1-p)$ and one can  use the
above theorem to find that 
\begin{align*}
	\pi_q^r(m,\omega)&=  \frac{1-\lambda}{1-p} \left(\frac{\lambda}{1-p} 
	\right)^{q},\\
	\pi_{m+1}^r(m,\omega)&=  \frac{(1-\lambda) \left(\frac{\lambda}{1-p} 
		\right)^{m+1} - p}{1-\lambda - p},
\end{align*}
for $q=0,\ldots,m$, which is in agreement with the closed form results presented
in \cite{tsitsiklis2013power}. Furthermore, for the case of exponential job sizes, one finds that 
$\omega = \frac{\lambda\pi_m^r(m,\omega)}{\pi_{m+1}^r(m,\omega)} - (1-p)$.

\subsection{Performance bounds}
In this section we investigate whether we can find bounds on the maximal queue length $\lceil \tilde m\rceil$. The next theorem shows that depending on $p$ and $\lambda$,  either the
maximum queue length equals one, rendering the model insensitive to the job size distribution, or $m$ can be made arbitrarily
large by varying the job size distribution, meaning there is no upper bound on $m$
that is valid for all PH job size distributions.

\begin{proposition}\label{prop:resourceperf}
	In the same setting as Theorem \ref{th:resource_pooling1}, we find that (in case $p < \lambda$) we have:
	\begin{itemize}
		\item In case $\frac{1}{1 + \lambda / (1-p)} >
		\frac{1-\lambda}{1-p}$ the maximum queue length is unbounded as a function of the job size distribution.
		\item Otherwise, the model is insensitive to the job size distribution and 
		the maximum queue length equals $1$.
	\end{itemize}
\end{proposition}
\begin{proof}
	First, we note that:
	$$
	\pi^r_0(1,\infty) = \frac{1}{1 + \lambda / (1-p)},
	$$	
	this shows that as long as $\frac{1-\lambda}{1-p} < \frac{1}{1 + \frac{\lambda}{1-p}}$, the maximal queue length is given by one.
	
	Otherwise, we take the PH distributions which we used to show Theorem \ref{th:mbounds}. That is, we define $Z(\varepsilon)$ as a PH distribution with transition matrix
	$$
	S = \begin{pmatrix}
		-\frac{(1-\varepsilon)}{\varepsilon} & 0\\
		0 & -\frac{\varepsilon}{1-\varepsilon}
	\end{pmatrix}
	$$
	and initial distribution $\alpha = (1-\varepsilon, \varepsilon)$. If we now fix $m \in \mathbb{N}$, we find that:
	$$
	\lim_{\varepsilon \rightarrow 0^+} \pi^r_0(m,\infty) = 1 / ( 1 + \lambda / (1-p) ) > \frac{1-\lambda}{1-p}.
	$$
	This shows that for any $m \in \mathbb{N}$ we can find an $\varepsilon > 0$ such that for $Z(\varepsilon)$ the maximal queue length exceeds $m$. This completes the proof.
\end{proof}

For the lower bound, we have the following result:
\begin{proposition}
		In the same setting as Theorem \ref{th:resource_pooling1} we find that the maximum queue length is minimized by having deterministic job sizes, while for PH distributions with $k$ phases the maximum queue length is minimized by having Erlang$-k$ job sizes.
Moreover, for deterministic job sizes, the maximal queue length  corresponds to the smallest $n \in \mathbb{N}$ for which:
		$$
		\frac{1}{1+ \rho \cdot \left( \sum_{k=0}^{n-1} \frac{(-1)^k}{k!} (n-1-k)^k e^{(n-1-k)\rho} \cdot \rho^k \right)} < \frac{1-\lambda}{1-p},
		$$
		with $\rho = \frac{\lambda}{1-p}$.
\end{proposition}
\begin{proof}
	For the first part, it was proven in \cite{miyazawa1990complementary} that the loss probability of an $M/G/1/m$ queue is increasing in the convex ordering. This shows 
	that the deterministic resp.~Erlang-$k$ distributions provide the smallest maximum queue lengths (for general resp.~$k$-phase job size distributions). The second part is a simple application of \cite{brun2000analytical}[Theorem 1] which presents a closed form formula
	for the probability that the M/D/1/m queue is idle.
\end{proof}

\subsection{Numerical Experiments}

In Figure \ref{fig:resourcev2} (left) we show the expected response time as a function of $\lambda$ for the resource pooling policy with the parameter setting $SCV=10$, $f=1/2$, $p \in \{0.2,0.3,0.4,0.5\}$.  As expected, decreasing $p$ or increasing $\lambda$ increases $m$. Further, as $1-\lambda$ decreases exponentially, $m$ seems to increase linearly. In other words, this example indicates a $\Theta(\log(1/(1-\lambda)))$ growth of the maximal queue length. 

In Figure \ref{fig:resourcev2} (right), we fix $p=0.25$, $f =1/2$ and use $\lambda \in \{0.8,0.85,0.9\}$ and $SCV \in [1,1000]$.
The figure clearly illustrates that $m$ is unbounded as a function of the $SCV$. 
Note that the mean response time of the resource pooling system does not exhibit non-differentiable points like the mean response times of the other systems. This is due to the fact that the central server always works on a pending job in a queue with maximum queue length (unless the
maximum queue length equals one).


\begin{figure}
    \centering
    \begin{minipage}{0.5\textwidth}
        \centering
        \includegraphics[width=1\textwidth]{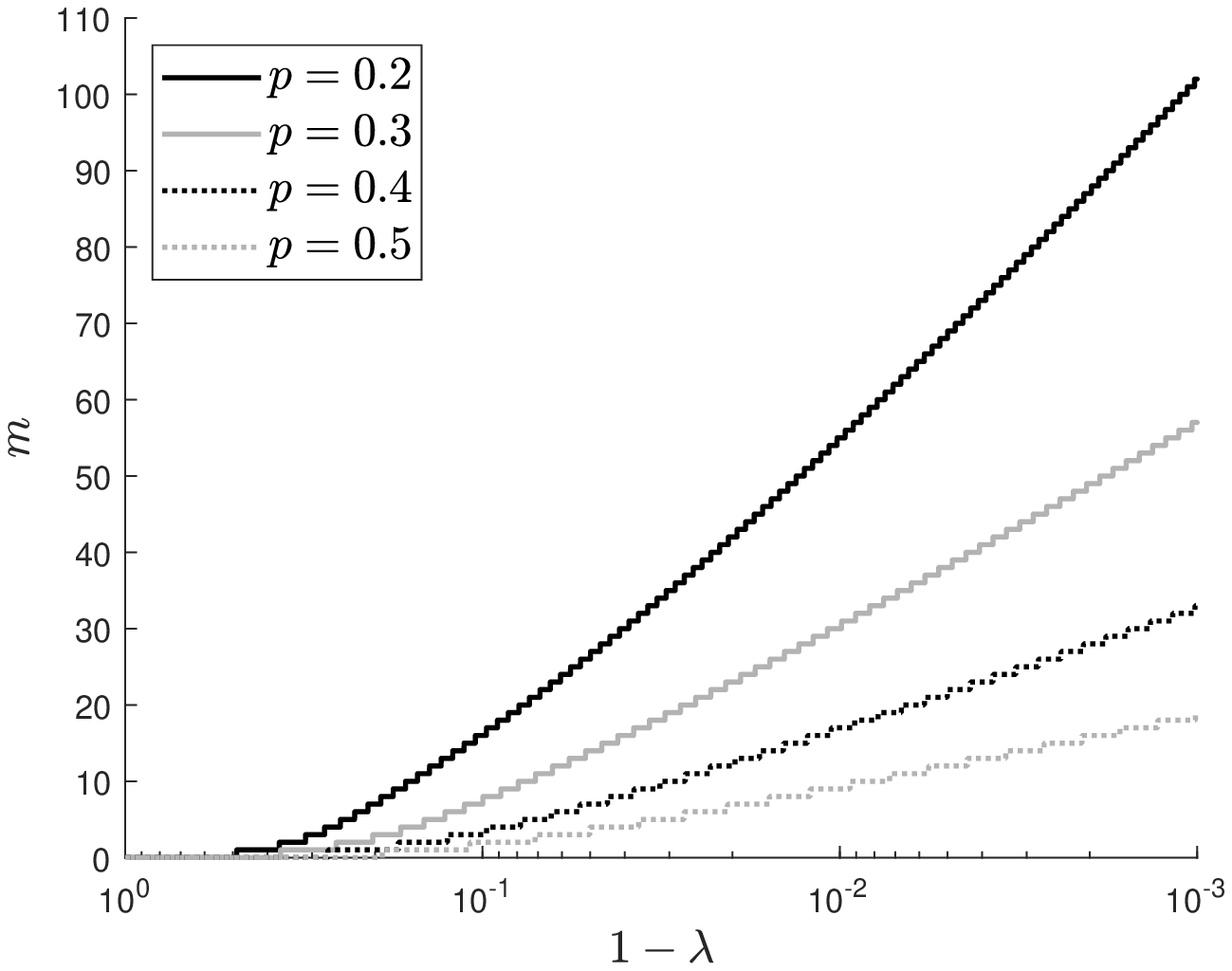}
    \end{minipage}\hfill
    \begin{minipage}{0.5\textwidth}
        \centering
        \includegraphics[width=1\textwidth]{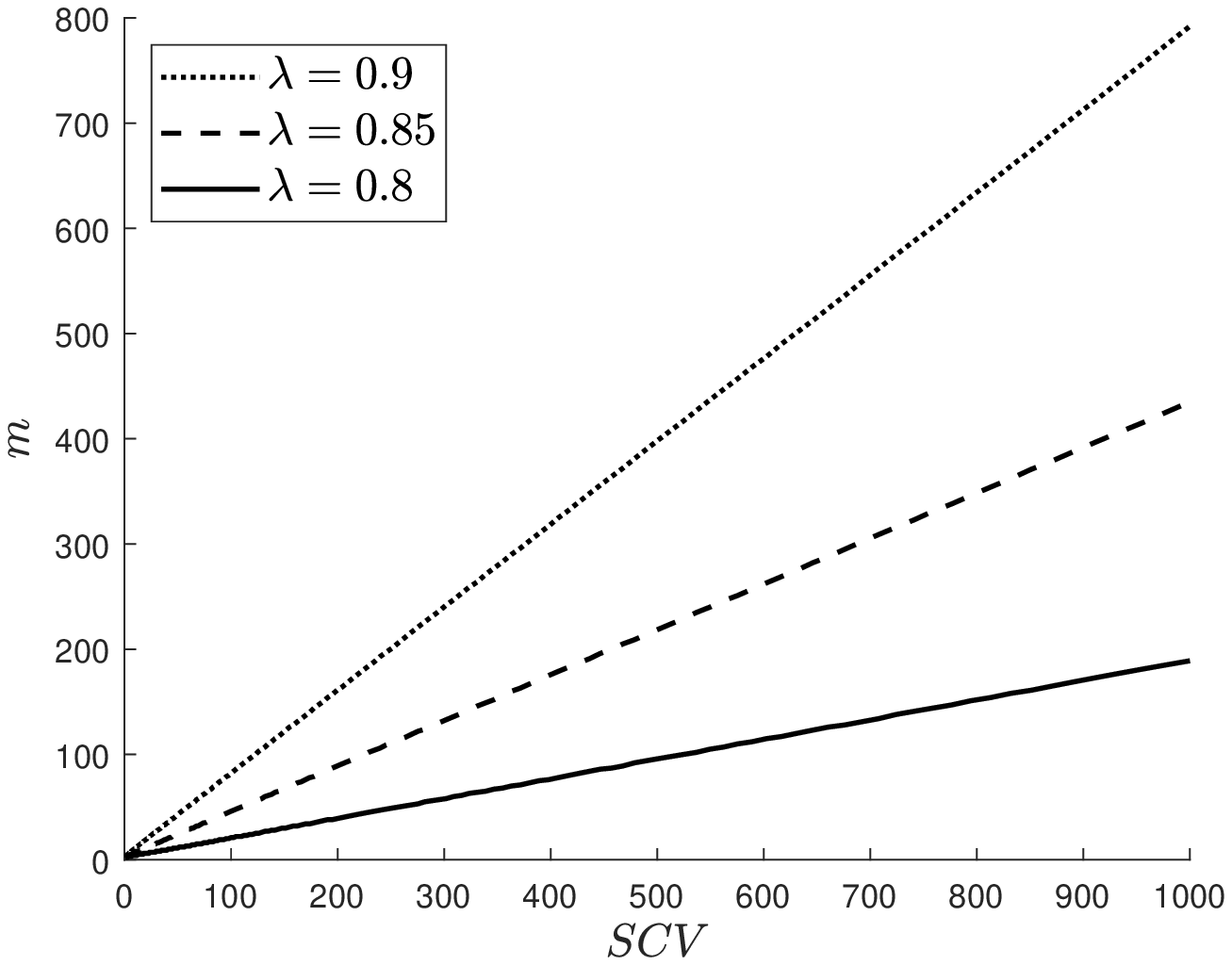}
    \end{minipage}
    \caption{Resource pooling: $m$ as a function of $\lambda$ for a variety of values of $p$ (left) and as a function function the $SCV$ for $3$ values of $\lambda$ (right).}\label{fig:resourcev2}
\end{figure}

\section{Simulation results} \label{app:simulation}
This section demonstrates that as $N$, that is the number of servers, becomes large 
the performance of the stochastic system seems to converge towards the
performance predicted by the cavity method. This suggests that the queue at the cavity
corresponds to the large-scale limit for the four policies considered. A formal proof
for this was presented in \cite{ying2017power} and \cite{tsitsiklis2013power} for
the water filling and resource sharing policies in the case of exponential job sizes.

For each policy we present simulation results for four different settings: a setting with exponential, hyperexponential, Erlang and Hyper-Erlang job sizes, with $N \in \{100,1000,10000,100000\}$. 

The job size distributions used for the experiments are all examples of PH distributions with mean 1. These are usually represented by $(\alpha,S)$, where $\alpha$ is the initial probability vector and $S$ a square matrix that records the rates of phase changes. The exponential distribution is obtained by setting $\alpha = 1$ and $S =-1$, while a hyperexponential distribution of order $2$ is found by setting $\alpha = (p, 1-p)$ for some probability $p$ and $S$ is a diagonal
matrix (with entries $-\mu_1$ and $-\mu_2$). A hyperexponential distribution of order 2 can be described using the mean of the distribution, the shape parameter $f$ and the squared coefficient of variation $SCV$, as indicated in \cite{hellemansSIG18}. The Erlang($k$) distribution is defined as a sum of $k$ exponential distributions (each with mean $1/k$), that is $\alpha_1 = 1$ and $S$ holds the values $-k$ on its main diagonal and $k$ on its upper diagonal. Let $S_{Erl(k)}$ denote the matrix $S$ of an Erlang($k$) distribution, the Hyper-Erlang($k,\ell$) distribution is then 
characterized by $\alpha_1 = p, \alpha_{k+1} = 1-p$ and 
\[S = \begin{bmatrix}
	S_{Erl(k)} & \\
	& S_{Erl(\ell)}
\end{bmatrix}.\]

The results presented  are only a small selection of the various settings we have simulated and 
are representative for other parameter settings as well. We ran the simulations starting from an empty system until $N\cdot 10^4$ arrivals  occurred, with a warm up period of 10\% of the jobs. The simulated average response times and 95\% confidence intervals are calculated based on 20 runs. The simulation results are given in Tables \ref{tab:simhyppush}-\ref{tab:simresourcepooling}.  The performance predicted by the cavity method is
found in the column  labelled $\infty$.

We see that in all considered cases the relative error typically decreases as $N$ increases,
with relative errors below $1\%$ for $N$ sufficiently large. We do note that for systems
of moderate size, e.g., $N=100$, the error can be substantial, exceeding $10\%$. 
We further note that while the cavity method often yields an optimistic prediction for
any finite $N$, this is not always the case here. This can be understood by noting that for a given arrival rate $\lambda$, we can set $\delta$ arbitrarily low such that the mean
response time is larger than the mean response time in an M/PH/1 queue, which corresponds to setting $N=1$. Hence, for $\delta$ small enough, the cavity method may yield a pessimistic prediction for finite $N$. Several such
examples can be seen for the push policy in Table \ref{tab:simhyppush}. 

Recall that for the water filling policy $M$ should grow as $\Theta(\log(N))$. In Table \ref{tab:simwaterfilling}, we set $M = C\cdot \log_{10}(N)$ (with the value of $C$ given in the table). Not surprisingly, we noted that the relative error of the cavity method  depends on the exact choice of the growth function.

\begin{table}[t]
	\centering
	\caption{Relative error of simulated mean response time for the hyperscalable push strategy based on 20 runs.}
	\label{tab:simhyppush}
	\begin{tabular}{ll:llllll}
		\toprule
		\multicolumn{2}{c:}{distribution} & \multicolumn{1}{c}{$\lambda$} & \multicolumn{1}{c}{$\delta$} &\multicolumn{1}{c}{$N$} &\multicolumn{1}{c}{sim. $\pm$ conf.} & \multicolumn{1}{c}{$\infty$} & \multicolumn{1}{c}{rel.err.\%} \\ \hline
		\multicolumn{2}{l:}{Exponential} & 0.9 &0.3& 100 & 5.8698 $\pm$ 2.11e-02 & 6.0081 &2.3028 \\
		\multicolumn{2}{l:}{} & 0.9 &0.3& 1000 & 6.0373 $\pm$ 6.70e-03 & 6.0081 & 0.4862\\
		& & 0.9 &0.3& 10000 & 6.0098 $\pm$ 1.39e-03 & 6.0081 & 0.0288\\
		& & 0.9 &0.3& 100000 & 6.0084 $\pm$ 6.30e-04 & 6.0081 & 0.0047\\\hline
		\multicolumn{2}{l:}{Hyperexponential(2)} & 0.85 &0.5& 100 & 4.7074 $\pm$ 4.67e-02 & 4.5862 & 2.6416\\
		\multicolumn{2}{l:}{$f=1/2, SCV =15$} & 0.85 &0.5& 1000 & 4.6229 $\pm$ 9.23e-03 & 4.5862 & 0.7996\\
		& & 0.85 &0.5& 10000 & 4.5877 $\pm$ 2.65e-03 & 4.5862 & 0.0314\\
		& & 0.85 &0.5& 100000 & 4.5867 $\pm$ 6.53e-04 & 4.5862 & 0.0106\\\hline
		\multicolumn{2}{l:}{Erlang(6)} & 0.8 &0.25& 100 & 4.0865 $\pm$ 1.02e-02 & 4.2206 & 3.1766\\
		& & 0.8 &0.25& 1000 & 4.2557 $\pm$ 6.43e-03 & 4.2206 & 0.8316\\
		& & 0.8 &0.25& 10000 & 4.2258 $\pm$ 1.71e-03 & 4.2206 & 0.1251\\
		& & 0.8 &0.25& 100000 & 4.2210 $\pm$ 4.53e-04 & 4.2206 & 0.0106\\\hline
		\multicolumn{2}{l:}{Hyper-Erlang(2,5)} & 0.85 &0.15& 100 & 7.9505 $\pm$ 1.77e-02 & 8.7304 & 8.9331\\
		\multicolumn{2}{l:}{$p=0.25$}& 0.85 &0.15& 1000 & 8.4868 $\pm$ 7.58e-03 & 8.7304 & 2.7905\\ 
		& & 0.85 &0.15& 10000 & 8.6962 $\pm$ 1.06e-03 & 8.7304 & 0.3923\\
		& & 0.85 &0.15& 100000 & 8.7266 $\pm$ 2.36e-04 & 8.7304 & 0.0431\\ \bottomrule
	\end{tabular}
\end{table}

\begin{table}[t]
	\centering
	\caption{Relative error of simulated mean response time for water filling strategy based on 20 runs.}
	\label{tab:simwaterfilling}
	\begin{tabular}{ll:llllllll}
		\toprule
		\multicolumn{2}{c:}{distribution} & \multicolumn{1}{c}{$\lambda$} & \multicolumn{1}{c}{$\delta$} &\multicolumn{1}{c}{$N$} &\multicolumn{1}{c}{$C$}&\multicolumn{1}{c}{$M$} &\multicolumn{1}{c}{sim. $\pm$ conf.} & \multicolumn{1}{c}{$\infty$} & \multicolumn{1}{c}{rel.err.\%} \\ \hline
		\multicolumn{2}{l:}{Exponential} & 0.8 &0.4& 100 & 20& 40 & 3.8973 $\pm$ 4.43e-02 & 3.5136 & 10.9205\\
		\multicolumn{2}{l:}{} & 0.8 &0.4& 1000 &20 & 60 & 3.5840 $\pm$ 1.49e-02 & 3.5136 & 2.0040\\
		& & 0.8 &0.4& 10000 &20 & 80 & 3.5446 $\pm$ 3.95e-03 & 3.5136 & 0.8812\\
		& & 0.8 &0.4& 100000 &20 & 100 & 3.5315 $\pm$ 1.68e-03 & 3.5136 & 0.5093\\\hline
		\multicolumn{2}{l:}{Hyperexponential(2)} & 0.8 &0.4 & 100 &40 & 80 & 5.5115 $\pm$ 1.09e-01 & 4.5947 & 19.9529 \\
		\multicolumn{2}{l:}{$f=1/2, SCV =10$} & 0.8 &0.4 & 1000 &40 &120 & 4.7841 $\pm$ 3.52e-02 & 4.5947 & 4.1217 \\
		& & 0.8 &0.4 & 10000 & 40& 160 & 4.6580 $\pm$ 9.20e-03 & 4.5947 & 1.3775\\
		& & 0.8 &0.4 & 100000 &40 & 200 & 4.6239 $\pm$ 2.61e-03 & 4.5947 & 0.6346\\\hline
		\multicolumn{2}{l:}{Erlang(3)} & 0.75 &1.2& 100&30 & 60 & 1.4877 $\pm$ 1.43e-02 & 1.4968 & 0.6059\\
		& & 0.75 &1.2& 1000 & 30& 90 & 1.5511 $\pm$ 6.14e-03 & 1.4968 & 3.6306\\
		& & 0.75 &1.2& 10000 & 30& 120 & 1.4975 $\pm$ 2.52e-03 & 1.4968 & 0.0502 \\
		& & 0.75 &1.2& 100000 & 30& 150 & 1.4963 $\pm$ 9.06e-04 & 1.4968 & 0.0298 \\\hline
		\multicolumn{2}{l:}{Hyper-Erlang(3,5)} & 0.8 &1.2& 100 &30 & 60 & 1.6386 $\pm$ 1.53e-02 & 1.5708 & 4.3178\\
		\multicolumn{2}{l:}{$p=0.6$} & 0.8 &1.2& 1000 & 30& 90 & 1.6993 $\pm$ 8.02e-03 & 1.5708 & 8.1847\\
		& & 0.8 &1.2& 10000 & 30& 120 & 1.5986 $\pm$ 2.74e-03 & 1.5708 &1.7696 \\
		& & 0.8 &1.2& 100000 & 30& 150 & 1.5756 $\pm$ 8.06e-04 & 1.5708 &0.3098 \\ \bottomrule
	\end{tabular}
\end{table}

\begin{table}[t]
	\centering
	\caption{Relative error of the simulated mean response time for the hyperscalable pull strategy based on 20 runs.}
	\label{tab:simhyppull}
	\begin{tabular}{ll:llllll}
		\toprule
		\multicolumn{2}{c:}{distribution} & \multicolumn{1}{c}{$\lambda$} & \multicolumn{1}{c}{$\delta$} &\multicolumn{1}{c}{$N$} &\multicolumn{1}{c}{sim. $\pm$ conf.} & \multicolumn{1}{c}{$\infty$} & \multicolumn{1}{c}{rel.err.\%} \\ \hline
		\multicolumn{2}{l:}{Exponential} & 0.7 &0.2& 100 & 2.0198 $\pm$ 3.70e-03 & 2.0816 & 2.9688\\
		\multicolumn{2}{l:}{} & 0.7 &0.2& 1000 & 2.0707 $\pm$ 1.28e-03 & 2.0816 &0.5237 \\
		& & 0.7 &0.2& 10000 & 2.0803 $\pm$ 3.78e-04 & 2.0816 & 0.0654\\
		& & 0.7 &0.2& 100000 & 2.0815 $\pm$ 9.90e-05 & 2.0816 & 0.0037\\\hline
		\multicolumn{2}{l:}{Hyperexponential(2)} & 0.9 &0.4& 100 & 2.5316 $\pm$ 4.76e-02 & 1.8726 & 35.1893\\
		\multicolumn{2}{l:}{$f=1/2, SCV =20$} & 0.9 &0.4& 1000 & 1.8590 $\pm$ 8.45e-03 & 1.8726 & 0.7271\\
		& & 0.9 &0.4& 10000 & 1.8540 $\pm$ 3.07e-03 & 1.8726 & 0.9965\\
		& & 0.9 &0.4& 100000 & 1.8711 $\pm$ 7.10e-04 & 1.8726 & 0.0836\\\hline
		\multicolumn{2}{l:}{Erlang(3)} & 0.75 &0.15& 100 & 2.6126 $\pm$ 6.58e-03 & 3.0000 & 12.9117\\
		& & 0.75 &0.15& 1000 & 2.7894 $\pm$ 2.75e-03 & 3.0000 & 7.0205\\
		& & 0.75 &0.15& 10000 & 2.8719 $\pm$ 3.54e-03 & 3.0000 & 4.2689\\
		& & 0.75 &0.15& 100000 & 2.9198 $\pm$ 4.00e-03 & 3.0000 & 2.6733\\\hline
		\multicolumn{2}{l:}{Hyper-Erlang(2,5)} & 0.75 &0.5& 100 & 1.2417 $\pm$ 1.99e-03 & 1.1839 & 4.8781\\
		\multicolumn{2}{l:}{$p=0.75$} & 0.75 &0.5& 1000 & 1.1888 $\pm$ 4.91e-04 & 1.1839 & 0.4126\\
		& & 0.75 &0.5& 10000 & 1.1845 $\pm$ 1.54e-04 & 1.1839 & 0.0536\\
		& & 0.75 &0.5& 100000 & 1.1839 $\pm$ 6.97e-05 & 1.1839 & 0.0038\\ \bottomrule
	\end{tabular}
\end{table}

\begin{table}[t]
	\centering
	\caption{Relative error of the simulated mean response time for the resource pooling strategy based on 20 runs.}
	\label{tab:simresourcepooling}
	\begin{tabular}{ll:llllll}
		\toprule
		\multicolumn{2}{c:}{distribution} & \multicolumn{1}{c}{$\lambda$} & \multicolumn{1}{c}{$p$} &\multicolumn{1}{c}{$N$} &\multicolumn{1}{c}{sim. $\pm$ conf.} & \multicolumn{1}{c}{$\infty$} & \multicolumn{1}{c}{rel.err.\%} \\ \hline
		\multicolumn{2}{l:}{Exponential} & 0.8 &0.3& 100 & 1.4774 $\pm$ 5.42e-03 & 1.3958 & 5.8454 \\
		\multicolumn{2}{l:}{} & 0.8 &0.3& 1000 & 1.4153 $\pm$ 1.27e-03 & 1.3958 & 1.4007\\
		& & 0.8 &0.3& 10000 & 1.3976 $\pm$ 5.90e-04 & 1.3958 & 0.1325\\
		& & 0.8 &0.3& 100000 & 1.3958 $\pm$ 2.35e-04 & 1.3958 & 0.0046\\\hline
		\multicolumn{2}{l:}{Hyperexponential(2)} & 0.7 &0.3& 100 & 1.0469 $\pm$ 7.06e-03 & 1.0699 & 2.1500 \\
		\multicolumn{2}{l:}{$f=1/2, SCV =5$} & 0.7 &0.3& 1000 & 1.0726 $\pm$ 1.59e-03 & 1.0699 & 0.2493\\
		& & 0.7 &0.3& 10000 & 1.0702 $\pm$ 4.63e-04 & 1.0699 & 0.0252\\
		& & 0.7 &0.3& 100000 & 1.0700 $\pm$ 1.91e-04 & 1.0699 & 0.0094\\\hline
		\multicolumn{2}{l:}{Erlang(7)} & 0.9 &0.5& 100 & 1.2995 $\pm$ 4.94e-03 & 1.2588 & 3.2315 \\
		& & 0.9 &0.5& 1000 & 1.2607 $\pm$ 1.33e-03 & 1.2588 & 0.1566\\
		& & 0.9 &0.5& 10000 & 1.2589 $\pm$ 4.05e-04 & 1.2588 & 0.0112\\
		& & 0.9 &0.5& 100000 & 1.2587 $\pm$ 1.01e-04 & 1.2588 & 0.0035\\\hline
		\multicolumn{2}{l:}{Hyper-Erlang(3,5)} & 0.8 &0.1& 100 & 2.0725 $\pm$ 4.47e-03 & 2.0320 & 1.9956 \\
		\multicolumn{2}{l:}{$p=0.6$} & 0.8 &0.1& 1000 & 2.0351 $\pm$ 1.66e-03 & 2.0320 & 0.1544\\
		& & 0.8 &0.1& 10000 & 2.0322 $\pm$ 3.88e-04 & 2.0320 & 0.0134\\
		& & 0.8 &0.1& 100000 & 2.0321 $\pm$ 1.51e-04 & 2.0320 & 0.0054\\ \bottomrule
	\end{tabular}
\end{table}

\section{Conclusion and future work} \label{sec:future_work}
Using the cavity approach, we studied four distinct load balancing policies which have a finite maximum queue length: the push \cite{van2019hyper}, water-filling
\cite{ying2017power}, pull and resource pooling \cite{tsitsiklis2013power} policies. 
Our main objective was to study the impact of the job size distribution as
prior work was limited to exponential job sizes. We found that in order to study the queue at the cavity for these policies two unknowns must be
determined: the maximum queue length and some rate or probability. 

For all the policies considered the maximum queue length can be studied using a simple
finite state Markov chain, often yielding closed form expressions (except for
resource pooling).  For most cases this maximum queue length scales as $\log\left( \frac{1}{1-\lambda} \right)$. The unknown rate or probability was determined next,
yielding an efficient way to compute the stationary distribution for the queue at the cavity.
Simulation results which show that the queue at the cavity corresponds to the large-scale
limit were presented in Section \ref{app:simulation}.

One significant pitfall of the push, pull and water filling policies is the fact that as $\delta$ decreases to zero, the maximum queue length increases to infinity (irrespective of the arrival rate $\lambda$). This entails that servers may suddenly receive many jobs in a short time period when their
queue length is updated. Interesting future work would be to 
adapt these policies to avoid such behavior.

The policies considered were studied in the context  of a single dispatcher. As the problem of having multiple dispatchers is becoming more and more relevant, one could try to generalize/adjust these policies in the presence of multiple dispatchers. Policies that operate  in such a setting have recently been studied in \cite{zhou2021asymptotically, vargaftik20}.



\bibliographystyle{ACM-Reference-Format}
\bibliography{thesis}


\begin{thebibliography}{28}


\ifx \showCODEN    \undefined \def \showCODEN     #1{\unskip}     \fi
\ifx \showDOI      \undefined \def \showDOI       #1{#1}\fi
\ifx \showISBNx    \undefined \def \showISBNx     #1{\unskip}     \fi
\ifx \showISBNxiii \undefined \def \showISBNxiii  #1{\unskip}     \fi
\ifx \showISSN     \undefined \def \showISSN      #1{\unskip}     \fi
\ifx \showLCCN     \undefined \def \showLCCN      #1{\unskip}     \fi
\ifx \shownote     \undefined \def \shownote      #1{#1}          \fi
\ifx \showarticletitle \undefined \def \showarticletitle #1{#1}   \fi
\ifx \showURL      \undefined \def \showURL       {\relax}        \fi
\providecommand\bibfield[2]{#2}
\providecommand\bibinfo[2]{#2}
\providecommand\natexlab[1]{#1}
\providecommand\showeprint[2][]{arXiv:#2}

\bibitem[\protect\citeauthoryear{Bramson, Lu, and Prabhakar}{Bramson
  et~al\mbox{.}}{2010}]%
        {bramsonLB}
\bibfield{author}{\bibinfo{person}{M. Bramson}, \bibinfo{person}{Y. Lu}, {and}
  \bibinfo{person}{B. Prabhakar}.} \bibinfo{year}{2010}\natexlab{}.
\newblock \showarticletitle{Randomized load balancing with general service time
  distributions}. In \bibinfo{booktitle}{\emph{{ACM SIGMETRICS} 2010}}.
  \bibinfo{pages}{275--286}.
\newblock
\urldef\tempurl%
\url{https://doi.org/10.1145/1811039.1811071}
\showDOI{\tempurl}


\bibitem[\protect\citeauthoryear{Bramson, Lu, and Prabhakar}{Bramson
  et~al\mbox{.}}{2012}]%
        {bramsonLB_QUESTA}
\bibfield{author}{\bibinfo{person}{M. Bramson}, \bibinfo{person}{Y. Lu}, {and}
  \bibinfo{person}{B. Prabhakar}.} \bibinfo{year}{2012}\natexlab{}.
\newblock \showarticletitle{Asymptotic independence of queues under randomized
  load balancing}.
\newblock \bibinfo{journal}{\emph{Queueing Syst.}} \bibinfo{volume}{71},
  \bibinfo{number}{3} (\bibinfo{year}{2012}), \bibinfo{pages}{247--292}.
\newblock
\urldef\tempurl%
\url{https://doi.org/10.1007/s11134-012-9311-0}
\showDOI{\tempurl}


\bibitem[\protect\citeauthoryear{Bramson, Lu, and Prabhakar}{Bramson
  et~al\mbox{.}}{2013}]%
        {bramsonAAP}
\bibfield{author}{\bibinfo{person}{M. Bramson}, \bibinfo{person}{Y. Lu}, {and}
  \bibinfo{person}{B. Prabhakar}.} \bibinfo{year}{2013}\natexlab{}.
\newblock \showarticletitle{Decay of tails at equilibrium for FIFO join the
  shortest queue networks}.
\newblock \bibinfo{journal}{\emph{Ann. Appl. Probab.}} \bibinfo{volume}{23},
  \bibinfo{number}{5} (\bibinfo{date}{10} \bibinfo{year}{2013}),
  \bibinfo{pages}{1841--1878}.
\newblock
\urldef\tempurl%
\url{https://doi.org/10.1214/12-AAP888}
\showDOI{\tempurl}


\bibitem[\protect\citeauthoryear{Brun and Garcia}{Brun and Garcia}{2000}]%
        {brun2000analytical}
\bibfield{author}{\bibinfo{person}{O. Brun} {and} \bibinfo{person}{J.-M.
  Garcia}.} \bibinfo{year}{2000}\natexlab{}.
\newblock \showarticletitle{Analytical solution of finite capacity M/D/1
  queues}.
\newblock \bibinfo{journal}{\emph{Journal of Applied Probability}}
  \bibinfo{volume}{37}, \bibinfo{number}{4} (\bibinfo{year}{2000}),
  \bibinfo{pages}{1092--1098}.
\newblock


\bibitem[\protect\citeauthoryear{Busic, Vliegen, and Scheller-Wolf}{Busic
  et~al\mbox{.}}{2012}]%
        {busicMOR}
\bibfield{author}{\bibinfo{person}{A. Busic}, \bibinfo{person}{I. Vliegen},
  {and} \bibinfo{person}{A. Scheller-Wolf}.} \bibinfo{year}{2012}\natexlab{}.
\newblock \showarticletitle{Comparing Markov Chains: Aggregation and Precedence
  Relations Applied to Sets of States, with Applications to Assemble-to-Order
  Systems}.
\newblock \bibinfo{journal}{\emph{Mathematics of Operations Research}}
  \bibinfo{volume}{37}, \bibinfo{number}{2} (\bibinfo{year}{2012}),
  \bibinfo{pages}{259--287}.
\newblock
\urldef\tempurl%
\url{https://doi.org/10.1287/moor.1110.0533}
\showDOI{\tempurl}


\bibitem[\protect\citeauthoryear{Delgado, Didona, Dinu, and Zwaenepoel}{Delgado
  et~al\mbox{.}}{2016}]%
        {delgado2016job}
\bibfield{author}{\bibinfo{person}{P. Delgado}, \bibinfo{person}{D. Didona},
  \bibinfo{person}{F. Dinu}, {and} \bibinfo{person}{W. Zwaenepoel}.}
  \bibinfo{year}{2016}\natexlab{}.
\newblock \showarticletitle{Job-aware scheduling in eagle: Divide and stick to
  your probes}. In \bibinfo{booktitle}{\emph{Proceedings of the Seventh ACM
  Symposium on Cloud Computing}}. \bibinfo{pages}{497--509}.
\newblock


\bibitem[\protect\citeauthoryear{Delgado, Dinu, Kermarrec, and
  Zwaenepoel}{Delgado et~al\mbox{.}}{2015}]%
        {delgado2015hawk}
\bibfield{author}{\bibinfo{person}{P. Delgado}, \bibinfo{person}{F. Dinu},
  \bibinfo{person}{A.-M. Kermarrec}, {and} \bibinfo{person}{W. Zwaenepoel}.}
  \bibinfo{year}{2015}\natexlab{}.
\newblock \showarticletitle{Hawk: Hybrid datacenter scheduling}. In
  \bibinfo{booktitle}{\emph{2015 $\{$USENIX$\}$ Annual Technical Conference
  ($\{$USENIX$\}$$\{$ATC$\}$ 15)}}. \bibinfo{pages}{499--510}.
\newblock


\bibitem[\protect\citeauthoryear{Gaver, Jacobs, and Latouche}{Gaver
  et~al\mbox{.}}{1984}]%
        {gaver1}
\bibfield{author}{\bibinfo{person}{D.P. Gaver}, \bibinfo{person}{P.A. Jacobs},
  {and} \bibinfo{person}{G. Latouche}.} \bibinfo{year}{1984}\natexlab{}.
\newblock \showarticletitle{Finite {Birth-and-Death} models in randomly
  changing environments}.
\newblock \bibinfo{journal}{\emph{Adv. in Appl. Probab.}}  \bibinfo{volume}{16}
  (\bibinfo{year}{1984}), \bibinfo{pages}{715--731}.
\newblock


\bibitem[\protect\citeauthoryear{G.H.~Hardy and Pólya}{G.H.~Hardy and
  Pólya}{1952}]%
        {hardy1952}
\bibfield{author}{\bibinfo{person}{J.E.~Littlewood G.H.~Hardy} {and}
  \bibinfo{person}{G. Pólya}.} \bibinfo{year}{1952}\natexlab{}.
\newblock \bibinfo{booktitle}{\emph{Inequalities}}.
\newblock \bibinfo{publisher}{2nd edition, Cambridge University Press}.
\newblock


\bibitem[\protect\citeauthoryear{Hellemans and Van~Houdt}{Hellemans and
  Van~Houdt}{2018}]%
        {hellemansSIG18}
\bibfield{author}{\bibinfo{person}{T. Hellemans} {and} \bibinfo{person}{B.
  Van~Houdt}.} \bibinfo{year}{2018}\natexlab{}.
\newblock \showarticletitle{On the Power-of-d-choices with Least Loaded Server
  Selection}.
\newblock \bibinfo{journal}{\emph{Proc. ACM Meas. Anal. Comput. Syst.}}
  (\bibinfo{date}{June} \bibinfo{year}{2018}).
\newblock
\showISSN{2476-1249}


\bibitem[\protect\citeauthoryear{Hellemans and Van~Houdt}{Hellemans and
  Van~Houdt}{2021}]%
        {hellemans2021mean}
\bibfield{author}{\bibinfo{person}{T. Hellemans} {and} \bibinfo{person}{B.
  Van~Houdt}.} \bibinfo{year}{2021}\natexlab{}.
\newblock \showarticletitle{Mean Waiting Time in Large-Scale and Critically
  Loaded Power of d Load Balancing Systems}.
\newblock \bibinfo{journal}{\emph{Proceedings of the ACM on Measurement and
  Analysis of Computing Systems}} \bibinfo{volume}{5}, \bibinfo{number}{2}
  (\bibinfo{year}{2021}), \bibinfo{pages}{1--34}.
\newblock


\bibitem[\protect\citeauthoryear{Kriege and Buchholz}{Kriege and
  Buchholz}{2014}]%
        {Kriege2014}
\bibfield{author}{\bibinfo{person}{J. Kriege} {and} \bibinfo{person}{P.
  Buchholz}.} \bibinfo{year}{2014}\natexlab{}.
\newblock \bibinfo{booktitle}{\emph{PH and MAP Fitting with Aggregated Traffic
  Traces}}.
\newblock \bibinfo{publisher}{Springer International Publishing},
  \bibinfo{address}{Cham}, \bibinfo{pages}{1--15}.
\newblock
\showISBNx{978-3-319-05359-2}
\urldef\tempurl%
\url{https://doi.org/10.1007/978-3-319-05359-2_1}
\showDOI{\tempurl}


\bibitem[\protect\citeauthoryear{Latouche and Ramaswami}{Latouche and
  Ramaswami}{1999}]%
        {latouche1}
\bibfield{author}{\bibinfo{person}{G. Latouche} {and} \bibinfo{person}{V.
  Ramaswami}.} \bibinfo{year}{1999}\natexlab{}.
\newblock \bibinfo{booktitle}{\emph{Introduction to Matrix Analytic Methods and
  stochastic modeling}}.
\newblock \bibinfo{publisher}{SIAM}, \bibinfo{address}{Philadelphia}.
\newblock


\bibitem[\protect\citeauthoryear{Maguluri and Srikant}{Maguluri and
  Srikant}{2016}]%
        {maguluri2016heavy}
\bibfield{author}{\bibinfo{person}{S.~T. Maguluri} {and} \bibinfo{person}{R.
  Srikant}.} \bibinfo{year}{2016}\natexlab{}.
\newblock \showarticletitle{Heavy traffic queue length behavior in a switch
  under the MaxWeight algorithm}.
\newblock \bibinfo{journal}{\emph{Stochastic Systems}} \bibinfo{volume}{6},
  \bibinfo{number}{1} (\bibinfo{year}{2016}), \bibinfo{pages}{211--250}.
\newblock


\bibitem[\protect\citeauthoryear{Mitzenmacher}{Mitzenmacher}{2001}]%
        {mitzenmacher2}
\bibfield{author}{\bibinfo{person}{M. Mitzenmacher}.}
  \bibinfo{year}{2001}\natexlab{}.
\newblock \showarticletitle{The Power of Two Choices in Randomized Load
  Balancing}.
\newblock \bibinfo{journal}{\emph{IEEE Trans. Parallel Distrib. Syst.}}
  \bibinfo{volume}{12} (\bibinfo{date}{October} \bibinfo{year}{2001}),
  \bibinfo{pages}{1094--1104}.
\newblock
Issue 10.
\showISSN{1045-9219}


\bibitem[\protect\citeauthoryear{Miyazawa}{Miyazawa}{1990}]%
        {miyazawa1990complementary}
\bibfield{author}{\bibinfo{person}{M. Miyazawa}.}
  \bibinfo{year}{1990}\natexlab{}.
\newblock \showarticletitle{Complementary generating functions for the
  MX/GI/1/k and GI/MY/1/k queues and their application to the comparison of
  loss probabilities}.
\newblock \bibinfo{journal}{\emph{Journal of applied probability}}
  \bibinfo{volume}{27}, \bibinfo{number}{3} (\bibinfo{year}{1990}),
  \bibinfo{pages}{684--692}.
\newblock


\bibitem[\protect\citeauthoryear{Neuts}{Neuts}{1981}]%
        {neuts2}
\bibfield{author}{\bibinfo{person}{M.F. Neuts}.}
  \bibinfo{year}{1981}\natexlab{}.
\newblock \bibinfo{booktitle}{\emph{Matrix-Geometric Solutions in Stochastic
  Models, An Algorithmic Approach}}.
\newblock \bibinfo{publisher}{John Hopkins University Press}.
\newblock


\bibitem[\protect\citeauthoryear{O'Cinneide}{O'Cinneide}{1991}]%
        {ocinneide2}
\bibfield{author}{\bibinfo{person}{C. O'Cinneide}.}
  \bibinfo{year}{1991}\natexlab{}.
\newblock \showarticletitle{Phase-type distributions and majorizations}.
\newblock \bibinfo{journal}{\emph{Annals of Applied Probability}}
  \bibinfo{volume}{1}, \bibinfo{number}{2} (\bibinfo{year}{1991}),
  \bibinfo{pages}{219--227}.
\newblock


\bibitem[\protect\citeauthoryear{Ousterhout, Wendell, Zaharia, and
  Stoica}{Ousterhout et~al\mbox{.}}{2013}]%
        {Sparrow}
\bibfield{author}{\bibinfo{person}{K. Ousterhout}, \bibinfo{person}{P.
  Wendell}, \bibinfo{person}{M. Zaharia}, {and} \bibinfo{person}{I. Stoica}.}
  \bibinfo{year}{2013}\natexlab{}.
\newblock \showarticletitle{Sparrow: Distributed, Low Latency Scheduling}. In
  \bibinfo{booktitle}{\emph{Proceedings of the Twenty-Fourth ACM Symposium on
  Operating Systems Principles}} \emph{(\bibinfo{series}{SOSP '13})}.
  \bibinfo{publisher}{ACM}, \bibinfo{address}{New York, NY, USA},
  \bibinfo{pages}{69--84}.
\newblock
\showISBNx{978-1-4503-2388-8}
\urldef\tempurl%
\url{https://doi.org/10.1145/2517349.2522716}
\showDOI{\tempurl}


\bibitem[\protect\citeauthoryear{Panchenko and Th\"{u}mmler}{Panchenko and
  Th\"{u}mmler}{2007}]%
        {panchenko1}
\bibfield{author}{\bibinfo{person}{A. Panchenko} {and} \bibinfo{person}{A.
  Th\"{u}mmler}.} \bibinfo{year}{2007}\natexlab{}.
\newblock \showarticletitle{Efficient Phase-type Fitting with Aggregated
  Traffic Traces}.
\newblock \bibinfo{journal}{\emph{Perform. Eval.}} \bibinfo{volume}{64},
  \bibinfo{number}{7-8} (\bibinfo{date}{Aug.} \bibinfo{year}{2007}),
  \bibinfo{pages}{629--645}.
\newblock
\showISSN{0166-5316}
\urldef\tempurl%
\url{https://doi.org/10.1016/j.peva.2006.09.002}
\showDOI{\tempurl}


\bibitem[\protect\citeauthoryear{Tsitsiklis and Xu}{Tsitsiklis and Xu}{2012}]%
        {tsitsiklis2012power}
\bibfield{author}{\bibinfo{person}{J.~N. Tsitsiklis} {and} \bibinfo{person}{K.
  Xu}.} \bibinfo{year}{2012}\natexlab{}.
\newblock \showarticletitle{On the power of (even a little) resource pooling}.
\newblock \bibinfo{journal}{\emph{Stochastic Systems}} \bibinfo{volume}{2},
  \bibinfo{number}{1} (\bibinfo{year}{2012}), \bibinfo{pages}{1--66}.
\newblock


\bibitem[\protect\citeauthoryear{Tsitsiklis and Xu}{Tsitsiklis and Xu}{2013}]%
        {tsitsiklis2013power}
\bibfield{author}{\bibinfo{person}{J.~N. Tsitsiklis} {and} \bibinfo{person}{K.
  Xu}.} \bibinfo{year}{2013}\natexlab{}.
\newblock \showarticletitle{On the power of (even a little) resource pooling}.
\newblock \bibinfo{journal}{\emph{Stochastic Systems}} \bibinfo{volume}{2},
  \bibinfo{number}{1} (\bibinfo{year}{2013}), \bibinfo{pages}{1--66}.
\newblock


\bibitem[\protect\citeauthoryear{van~der Boor, Borst, and van
  Leeuwaarden}{van~der Boor et~al\mbox{.}}{2019}]%
        {van2019hyper}
\bibfield{author}{\bibinfo{person}{Mark van~der Boor}, \bibinfo{person}{Sem
  Borst}, {and} \bibinfo{person}{Johan van Leeuwaarden}.}
  \bibinfo{year}{2019}\natexlab{}.
\newblock \showarticletitle{Hyper-scalable JSQ with sparse feedback}.
\newblock \bibinfo{journal}{\emph{Proceedings of the ACM on Measurement and
  Analysis of Computing Systems}} \bibinfo{volume}{3}, \bibinfo{number}{1}
  (\bibinfo{year}{2019}), \bibinfo{pages}{1--37}.
\newblock


\bibitem[\protect\citeauthoryear{van~der Boor, Borst, and van
  Leeuwaarden}{van~der Boor et~al\mbox{.}}{2021}]%
        {van2021optimal}
\bibfield{author}{\bibinfo{person}{Mark van~der Boor}, \bibinfo{person}{Sem
  Borst}, {and} \bibinfo{person}{Johan van Leeuwaarden}.}
  \bibinfo{year}{2021}\natexlab{}.
\newblock \showarticletitle{Optimal hyper-scalable load balancing with a strict
  queue limit}.
\newblock \bibinfo{journal}{\emph{Performance Evaluation}}
  (\bibinfo{year}{2021}), \bibinfo{pages}{102217}.
\newblock


\bibitem[\protect\citeauthoryear{Vargaftik, Keslassy, and Orda}{Vargaftik
  et~al\mbox{.}}{2020}]%
        {vargaftik20}
\bibfield{author}{\bibinfo{person}{S. Vargaftik}, \bibinfo{person}{I.
  Keslassy}, {and} \bibinfo{person}{A. Orda}.} \bibinfo{year}{2020}\natexlab{}.
\newblock \showarticletitle{LSQ: Load Balancing in Large-Scale Heterogeneous
  Systems With Multiple Dispatchers}.
\newblock \bibinfo{journal}{\emph{IEEE/ACM Trans. Netw.}} \bibinfo{volume}{28},
  \bibinfo{number}{3} (\bibinfo{date}{June} \bibinfo{year}{2020}),
  \bibinfo{pages}{1186–1198}.
\newblock
\showISSN{1063-6692}
\urldef\tempurl%
\url{https://doi.org/10.1109/TNET.2020.2980061}
\showDOI{\tempurl}


\bibitem[\protect\citeauthoryear{Vvedenskaya, Dobrushin, and
  Karpelevich}{Vvedenskaya et~al\mbox{.}}{1996}]%
        {vvedenskaya3}
\bibfield{author}{\bibinfo{person}{N.D. Vvedenskaya}, \bibinfo{person}{R.L.
  Dobrushin}, {and} \bibinfo{person}{F.I. Karpelevich}.}
  \bibinfo{year}{1996}\natexlab{}.
\newblock \showarticletitle{Queueing System with Selection of the Shortest of
  Two Queues: an Asymptotic Approach}.
\newblock \bibinfo{journal}{\emph{Problemy Peredachi Informatsii}}
  \bibinfo{volume}{32} (\bibinfo{year}{1996}), \bibinfo{pages}{15--27}.
\newblock


\bibitem[\protect\citeauthoryear{Ying, Srikant, and Kang}{Ying
  et~al\mbox{.}}{2017}]%
        {ying2017power}
\bibfield{author}{\bibinfo{person}{L. Ying}, \bibinfo{person}{R. Srikant},
  {and} \bibinfo{person}{X. Kang}.} \bibinfo{year}{2017}\natexlab{}.
\newblock \showarticletitle{The power of slightly more than one sample in
  randomized load balancing}.
\newblock \bibinfo{journal}{\emph{Mathematics of Operations Research}}
  \bibinfo{volume}{42}, \bibinfo{number}{3} (\bibinfo{year}{2017}),
  \bibinfo{pages}{692--722}.
\newblock


\bibitem[\protect\citeauthoryear{Zhou, Shroff, and Wierman}{Zhou
  et~al\mbox{.}}{2021}]%
        {zhou2021asymptotically}
\bibfield{author}{\bibinfo{person}{X. Zhou}, \bibinfo{person}{N. Shroff}, {and}
  \bibinfo{person}{A. Wierman}.} \bibinfo{year}{2021}\natexlab{}.
\newblock \showarticletitle{Asymptotically optimal load balancing in
  large-scale heterogeneous systems with multiple dispatchers}.
\newblock \bibinfo{journal}{\emph{Performance Evaluation}}
  \bibinfo{volume}{145} (\bibinfo{year}{2021}), \bibinfo{pages}{102146}.
\newblock


\end{thebibliography}

\end{document}